%% file: main.tex
\documentclass[12pt]{article}

\usepackage[margin=1in]{geometry}
\usepackage{amsmath,amssymb,amsthm,mathtools}
\usepackage{booktabs}
\usepackage{graphicx}
\usepackage{xcolor}
\usepackage{placeins}
\usepackage{microtype}
\usepackage[colorlinks=true,linkcolor=blue,citecolor=blue,urlcolor=blue,hyperfootnotes=false]{hyperref}
\usepackage[authoryear]{natbib}
\usepackage{setspace}
\newtheorem{theorem}{Theorem}
\newtheorem{proposition}{Proposition}
\newtheorem{corollary}{Corollary}
\newtheorem{lemma}{Lemma}
\theoremstyle{remark}

\newcommand{\R}{\mathbb{R}}
\newcommand{\E}{\mathbb{E}}
\newcommand{\Var}{\mathrm{Var}}
\newcommand{\Pbb}{\mathbb{P}}
\newcommand{\norm}[1]{\left\lVert #1 \right\rVert}
\newcommand{\risk}{\mathcal{R}}

\title{Variance or Standard Deviation? Shell Geometry and Global-Scale Priors in High-Dimensional Shrinkage}
\author{Wayne Yuan Gao\thanks{Department of Economics, University of Pennsylvania, 
 waynegao@upenn.edu.}\,\, and Zhiheng You\thanks{Department of Economics, University of Pennsylvania, 
 zhyou@sas.upenn.edu.}}

\begin{document}
\maketitle

\begin{abstract}
\noindent We study how the choice of default prior for a common Gaussian scale affects high-dimensional shrinkage risk, highlighting the role played by high-dimensional geometry. Formally, we consider a high-dimensional setting in which the near-zero behavior of the common scale prior has first-order consequences for shrinkage risk, and show that priors that are flat on the variance and those flat on the standard deviation allocate markedly different mass near the zero-scale boundary, leading to distinct shrinkage behavior and informing principled default prior selection. Specifically, under a radial-power benchmark, we establish that the SD-flat benchmark has a one-unit asymptotic risk advantage near the origin, crosses over in the critical regime, and is second-order equivalent to the variance-flat benchmark for strong signals. Proper single global-scale hyperpriors and bounded coordinate-multiplier mixtures inherit these limits through the near-zero exponent of their SD-scale density. For heavier-tailed or sparse priors, that exponent still classifies the common global-scale component, while local-scale tails, model-size priors, or allocation priors can also affect risk.\\

\end{abstract}

\newpage

\section{Introduction}

A recurring default choice in hierarchical Bayes and high-dimensional shrinkage is how to place a prior on a common global Gaussian scale. Should the benchmark default be flat on the variance component or flat on the standard deviation? The two formulations are often treated as nearby modeling choices, but in high dimensions, they induce different scale geometries near the origin. The question is whether these different ways of weighting scales have systematically different frequentist risk consequences in high-dimensional settings, and how this should inform proper prior choices in practice.

We study these questions in the Gaussian normal means problem through a radial-power benchmark that nests the two canonical defaults: one corresponding to Stein's harmonic prior matches a flat prior on the variance, while the other matches a flat prior on the standard deviation. We analyze the generalized Bayes posterior mean across three asymptotic regimes for signal energy, meaning the squared Euclidean size of the unknown mean vector, as the dimension $d$ increases to infinity: weak-signal sequences with bounded energy, a critical regime in which signal energy is on the natural $\sqrt d$ fluctuation scale of Gaussian noise, and strong-signal sequences with energy of order $d$.

High dimension is not only of increasing relevance in applied work, but also provides a theoretical framework that makes the prior-scale question geometrically sharp. In high-dimensional Euclidean space, volume is concentrated in spherical shells: in polar coordinates the volume element contains the factor $r^{d-1}$. Here and below, a shell means a spherical shell or thin annulus: a set of points whose Euclidean norm lies in a small interval. Likewise, a standard Gaussian noise vector in $d$ dimensions typically lies in a thin annulus around radius $\sqrt d$; equivalently, its squared norm is $d$ plus random fluctuations of order $\sqrt d$. Such concentration of measure is one of the central lessons of high-dimensional geometry: mass can concentrate on sets that look negligible from low-dimensional intuition, and small changes in radial parametrization can become first-order under the relevant near-zero scaling \citep{Ball1997,Ledoux2001,Vershynin2018}. Related geometric effects also underlie phase-transition phenomena in high-dimensional data analysis and signal processing \citep{DonohoTanner2009}. 

In the Gaussian normal-means experiment considered here, rotational invariance makes this geometry explicit. The benchmark posterior depends on the data through their squared Euclidean norm, and the asymptotic regimes are organized by the squared Euclidean size of the unknown mean. The two defaults (variance-flat versus SD-flat) differ exactly in how they weight radial shells near the zero-scale boundary. The high-dimensional limit therefore converts a seemingly local reparameterization choice (over variance or SD) into a quantitative risk comparison.

Our main results are as follows. In the weak regime, the SD-flat default improves on the variance-flat benchmark by one asymptotic risk unit. In the critical regime, we obtain an explicit limit for the centered risk and show that the SD-scale advantage holds at lower critical-regime signal strengths but can reverse as signal strength within that regime grows. In the strong regime, the risk expansion is universal across fixed radial exponents up to second order, and thus the two defaults yield asymptotically equivalent risk. 

We then prove that these benchmark results transfer to a much wider range of settings, formalizing the heuristic that the key tradeoff in the choice of variance vs SD is driven by the high-dimensional scale geometry rather than by the particular radial-power benchmark. First, proper single global-scale hyperpriors inherit the same low-signal behavior through the near-zero exponent of their global SD density. Second, the same conclusion holds for bounded coordinate-multiplier normal mixtures. This class contains bounded global--local normal scale mixtures when the multipliers are positive continuous local scales, and Bernoulli spike-and-slab priors with fixed inclusion probability when the multipliers are binary indicators. The theorem separates the contribution of the common scale from additional mechanisms (heavy local-scale tails, dimension-dependent Dirichlet weight vectors, and model-size priors) that govern risk for common sparse priors outside the bounded-multiplier class.

The benchmark comparison belongs to the classical normal-means shrinkage literature because the benchmark rules are generalized Bayes estimators in the same radial family studied by \citet{Stein1956,Brown1971,Strawderman1971,MaruyamaTakemura2008,BrownZhao2012}. Our emphasis is different, however. We do not study admissibility, minimaxity, or other fixed-dimensional decision-theoretic properties. The object here is frequentist risk under high dimension ($d\to\infty$) asymptotics. On the hierarchical-Bayes side, \citet{Gelman2006} emphasizes that default priors for variance components are not innocuous under reparameterization and argues for working on the standard-deviation scale, leading to the half-$t$ family as a default for hierarchical variance parameters. \citet{PolsonScott2012} single out the half-Cauchy as a useful prior for a common global scale, highlighting its heavy tails, substantial mass near zero, and convenient scale-mixture representation. \citet{PiironenVehtari2017Hyperprior,PiironenVehtari2017Sparsity} make a complementary point for the horseshoe: the numerical scale of the global parameter can be calibrated from prior sparsity information, for example through the induced effective number of nonzero coefficients. The present analysis addresses a different component of the prior specification: it does not choose the numerical scale, but identifies the near-zero SD-scale exponent that controls high-dimensional low-signal risk. Thus the prior shape near zero and the numerical scale calibration are distinct components of the global-scale specification.

The comparison also connects to the modern shrinkage literature motivated by sparsity. Global--local priors such as the horseshoe and its variants were developed for nearly-black signals and related sparse normal means problems \citep*{CarvalhoPolsonScott2010,vanDerPasKleijnvanDerVaart2015,BhadraDattaPolsonWillard2016}. Spike-and-slab priors give a second major route to sparse shrinkage and variable selection: the standard specification introduces Bernoulli inclusion indicators, a model-size prior or inclusion-probability prior, and a slab distribution for included coefficients. Fully Bayes and empirical Bayes treatments of the inclusion probability provide multiplicity adjustment in high-dimensional model selection \citep{ScottBerger2010,CastilloSzabo2020}. \citet{Rockova2018Continuous} studies continuous spike-and-slab priors in sparse normal means, while \citet{RockovaGeorge2018SSL} develop the spike-and-slab LASSO for high-dimensional regression with adaptive nonseparable penalties and multiplicity adjustment. In the regression setting with unknown noise level, \citet{MoranRockovaGeorge2019Variance} show that variance-prior form can materially affect high-dimensional Bayesian variable selection. These papers are essential background, but they answer a different question from the one studied here. We deliberately average over directions and focus on the common global-scale component, because many applications feature many weak effects rather than a clean split between a few large coordinates and exact zeros. This viewpoint is compatible with sequence-space work outside exact sparsity \citep{JohnstoneSilverman2004}, with hybrid dense-plus-sparse formulations \citep{ChernozhukovHansenLiao2017}, and with recent econometric arguments that exact sparsity can be empirically fragile \citep{GiannoneLenzaPrimiceri2021,KolesarMullerRoelsgaard2025}.

The practical message is therefore more transferable than a mere ranking of named priors. We identify the near-zero exponent of the common shrinkage scale as the quantity that governs low-signal isotropic risk when the remaining prior components preserve the same local near-zero asymptotic problem. For heavy-tailed global--local priors, Dirichlet--Laplace and R2-D2 shrinkage priors, and sparse model-selection priors, this exponent remains a useful specification of the common scale, while the full frequentist risk also depends on the additional component that defines the prior architecture.

The rest of the paper is organized as follows. Section~\ref{sec:model} introduces the benchmark model, the scale-mixture representation, the three asymptotic risk regimes, and our main results. Section~\ref{sec:proper} transfers the benchmark results to proper priors, to common single global-scale hyperpriors, and to bounded coordinate-multiplier priors. Section~\ref{sec:sim} gives numerical illustrations, and Section~\ref{sec:discussion} concludes. The appendices contain proofs and auxiliary calculations.

\section{Benchmark Setup and Main Results}
\label{sec:model}

\subsection{Setup}
\label{sec:setup}
Consider the following normal means experiment
\[
X_d\sim N_d(\theta_d,I_d),
\]
under a high-dimensional setting where $d\to\infty$. We estimate the normal mean vector $\theta_d\in\R^d$ by shrinkage rules. In the formulation used below, the shrinkage rules are generalized Bayes posterior means induced by appropriate priors on $\theta_d$, so the choice of prior determines the corresponding shrinkage rule.\footnote{We use ``generalized Bayes'' in the standard decision-theoretic sense: the formal posterior mean is computed from an improper prior when the marginal integral is finite, so the resulting estimator need not arise from a proper probability prior. Generalized Bayes shrinkage rules based on improper radial priors are classical in the normal-means problem; see, for example, \citet{Brown1971,Strawderman1971} and \citet{BergerStrawdermanTang2005}.}

Our benchmark family of priors is the (improper) radial prior
\[
\pi_{d,c}(\theta)\propto \norm{\theta}^{-(d-c)},
\quad c>0,
\]
which nests the two commonly used global-scale geometries of interest. The family is a natural benchmark because rotational invariance removes directional modeling choices, and the polar-coordinate calculation below shows that the single exponent $c$ controls the radial-coordinate weighting. It therefore isolates the global-scale coordinate question without mixing it with sparsity, tail, or coordinate-selection mechanisms.

In polar coordinates $\theta=r\omega$, the prior $\pi_{d,c}(\theta)\propto\norm{\theta}^{-(d-c)}$ induces marginal measure
\[
r^{c-1}\,dr
\]
for the radial coordinate $r=\norm{\theta}$. Hence $c=1$ is flat in $r$, corresponding to the SD-flat benchmark. Equivalently, it gives equal mass to spherical shells of equal radial thickness. By contrast, $c=2$ is flat in $r^2=\norm{\theta}^2$, corresponding to the variance-flat benchmark.

Alternatively, we may also view the radial prior from the following perspective.

\begin{proposition}[Scale-mixture representation]
\label{prop:mixture}
For $c>0$ and $d>c$,
\[
\norm{\theta}^{-(d-c)}
\propto
\int_0^\infty N_d(\theta;0,gI_d)\,g^{c/2-1}\,dg.
\]
\end{proposition}
Proposition~\ref{prop:mixture} implies that $\pi_{d,c}$ is the marginal prior induced by the conjugate hierarchy $\theta\mid g\sim N_d(0,gI_d)$ with hyperprior $p(g)\propto g^{c/2-1}$ on the global variance component $g$. 

Hence, the two defaults of interest are:
\begin{itemize}
 \item Variance-Flat: $$c=2 \quad\iff\quad p(g)\propto 1.$$
 \item SD-Flat: $$c=1 \quad\iff\quad p(\tau)\propto 1,\text{ where } \tau:=\sqrt g$$
\end{itemize}

Both defaults have been well studied in previous literature. The variance-flat choice is algebraically conjugate and, through Proposition~\ref{prop:mixture}, coincides with the harmonic radial prior that anchors much of the classical Stein-shrinkage literature \citep{Stein1956,Brown1971,Strawderman1971,MaruyamaTakemura2008,BrownZhao2012}. The SD-flat choice is the default advocated in hierarchical variance-component modeling when the scale itself is the interpretable quantity; it is also the near-zero behavior of half-normal, half-$t$, and half-Cauchy scale priors with positive density at zero \citep{Gelman2006,PolsonScott2012}.

\medskip

Motivated by this radial scale geometry, we compare the two defaults by the frequentist squared-error risk of their generalized Bayes posterior means. For an estimator $\delta$, define
\[
\risk(\theta,\delta)=\E_\theta \norm{\delta(X)-\theta}^2.
\]
For later use, write $T_d=\norm{X_d}^2$ for the quadratic statistic, or equivalently the observed squared norm, and write $Y=g/(1+g)\in(0,1)$ for the variance-mixture shrinkage factor. We refer to $\norm{\theta_d}$ as the signal norm and to $\norm{\theta_d}^2$ as the signal energy. The generalized Bayes posterior mean under $\pi_{d,c}$ is the isotropic shrinkage rule
\[
\delta_{d,c}(x)=s_{d,c}(\norm{x}^2)x,
\quad
s_{d,c}(t)=\E[Y\mid T_d=t].
\]

The relevant asymptotic scale is the scale of the quadratic statistic. The same Gaussian sequence model underlies much of the high-dimensional shrinkage and sparse-normal-means literature, where regimes are often defined by sparsity or by signal energy rather than by fixed dimension alone.\footnote{For sparsity-based regimes in the normal-means problem, one typically writes $s_d=|\{j:\theta_{j,d}\ne0\}|$ and distinguishes nearly-black or sparse classes, such as $s_d=o(d)$, from dense classes with $s_d\asymp d$; see \citet{JohnstoneSilverman2004,CarvalhoPolsonScott2010,vanDerPasKleijnvanDerVaart2015}. Detection formulations instead parameterize alternatives by an $\ell_2$ radius or total signal energy relative to Gaussian noise; see \citet{IngsterSuslina2000}.}

We study three asymptotic regimes for the signal energy: as $d\to\infty$,
\[
\begin{aligned}
\text{Weak Regime: }&\ \norm{\theta_d}^2\to \nu,\\
\text{Critical Regime: }&\ \frac{\norm{\theta_d}^2}{\sqrt d}\to \beta,\\
\text{Strong Regime: }&\ \norm{\theta_d}^2=\rho d+\kappa+o(1).
\end{aligned}
\]
The weak regime has bounded total energy, the critical regime has signal energy on the null fluctuation scale of $T_d$, and the strong regime has energy proportional to dimension.

The three regimes are tied to fluctuations of the quadratic statistic rather than chosen ad hoc. When $\theta_d=0$, $T_d\sim\chi^2_d$, so $\E T_d=d$ and $\Var(T_d)=2d$, and $T_d-d$ is on the scale of $O(\sqrt d)$ with high probability \citep{LaurentMassart2000}, consistent with the Gaussian-annulus viewpoint in high-dimensional probability \citep{Ledoux2001,Vershynin2018}. A bounded signal energy (in the weak regime) is therefore below the intrinsic null fluctuation scale, $\norm{\theta_d}^2\asymp \sqrt d$ is the critical scaling at which the signal shifts $T_d$ by the same order as noise, and $\norm{\theta_d}^2\asymp d$ changes the normalized value $T_d/d$ at leading order.

\medskip

The following two lemmas are useful for the frequentist-risk calculations. Lemma~\ref{prop:beta} reduces the posterior shrinkage factor to a one-dimensional conditional distribution given the quadratic statistic, and Lemma~\ref{lem:riskid} rewrites squared-error risk in terms of that shrinkage factor.

\begin{lemma}[Posterior representation]
\label{prop:beta}
If $d>c$, then conditional on $T_d=t$ the posterior density of $Y$ is proportional to
\[
y^{c/2-1}(1-y)^{(d-c)/2-1} e^{ty/2}, \quad 0<y<1.
\]
\end{lemma}

\begin{lemma}[Exact risk identity]
\label{lem:riskid}
For every $d>c$,
\[
\risk(\theta_d,\delta_{d,c})
=
\norm{\theta_d}^2 + 2c - \E\!\left[T_d s_{d,c}(T_d)^2\right].
\]
\end{lemma}

\subsection{Weak Regime}
\label{sec:weak}

We start with the weak regime, where the total signal energy remains bounded. In this regime the posterior mean is evaluated close to the origin, where the radial exponent $c$ directly controls the risk constant.

\begin{theorem}[Weak-signal limit]
\label{thm:bounded}
If $\norm{\theta_d}^2\to \nu<\infty$, then
\[
\risk(\theta_d,\delta_{d,c})\to \nu+c.
\]
In particular,
\[
\risk(\theta_d,\delta_{d,1})-\risk(\theta_d,\delta_{d,2})\to -1.
\]
\end{theorem}

Thus the SD-flat benchmark has a one-unit asymptotic risk advantage over the variance-flat benchmark in the weak regime.

Theorem~\ref{thm:bounded} gives the first main message of the paper: in the weak regime, the SD-flat prior is preferable to the variance-flat one as measured by the asymptotic risk $\mathcal{R}$. This complements the argument of \citet{Gelman2006}, who recommends working on the standard-deviation scale because the standard deviation is directly interpretable in the original model, the improper uniform density on that scale can be understood as a limit of proper half-$t$ priors, and common alternatives such as inverse-gamma priors or a flat prior on the variance can be sensitive near zero or miscalibrated toward larger variance components. It also clarifies how the half-Cauchy recommendation of \citet{PolsonScott2012} enters the present comparison: their case for the half-Cauchy is that it is a proper top-level scale prior with substantial mass near zero and heavy tails, giving good frequentist risk near the origin without severe compromises elsewhere. In our benchmark, its positive finite density at $\tau=0$ places its common-scale behavior in the same near-zero $c=1$ class as the SD-flat default.

Theorem~\ref{thm:bounded} is the bottom endpoint of the critical-regime analysis: a bounded signal energy lies below the $\sqrt d$ null-fluctuation scale of $T_d$, so it does not alter the limiting geometry of the statistic, and the result follows from the proof of Theorem~\ref{thm:main}. The key proof heuristic is therefore deferred until after Theorem \ref{thm:main}. Here we only briefly explain why the limit under the weak regime is so clean and simple. In this regime, given the weakness of the signal, the only relevant feature of the prior  is its near-zero exponent $c$: the centered risk converges to a constant that a Gaussian integration-by-parts identity (established for the critical regime) collapses to exactly $c$, with every other term vanishing. The total energy $\norm{\theta_d}^2\to\nu$ then simply adds back, because the estimator barely shrinks this close to the origin, giving the simple additive form $\nu+c$ and the immediate one-unit gap between the two benchmarks.

\begin{figure}[h!]
\caption{Weak-Regime Excess Risk}
\label{fig:weak-benchmark-main}
\begin{center}
\includegraphics[width=0.78\textwidth,trim=0 0 486.7pt 0,clip]{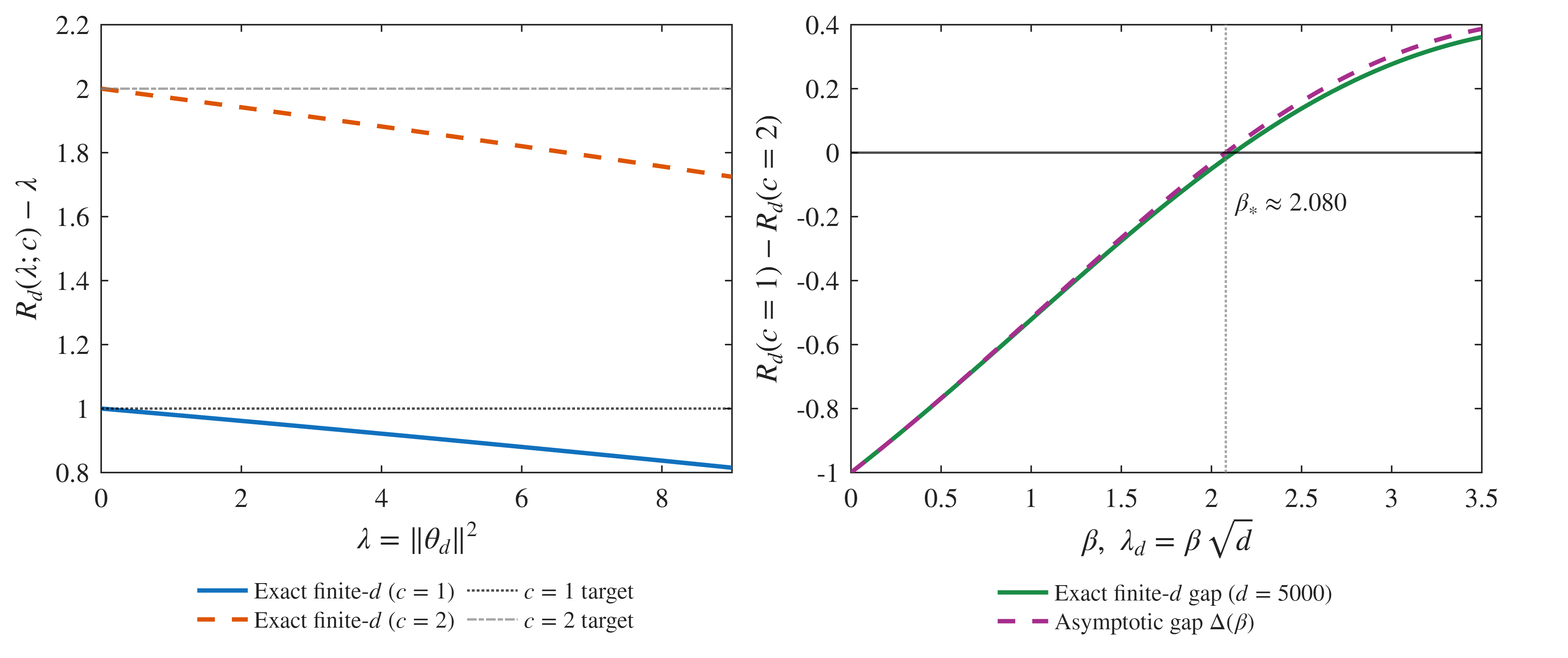}
\end{center}
{\footnotesize {\em Notes}: Weak-regime excess risk $\risk(\theta_d,\delta_{d,c})-\lambda_d$ at $d=2000$, where $\lambda_d=\norm{\theta_d}^2$, with horizontal asymptotic targets $c=1$ and $c=2$.}
\setlength{\baselineskip}{4mm}
\end{figure}

We confirm this asymptotic risk gap numerically. Figure~\ref{fig:weak-benchmark-main} plots the excess risk $\risk(\theta_d,\delta_{d,c})-\lambda_d$, with $\lambda_d=\norm{\theta_d}^2$, for the weak-signal calibration at $d=2000$. The horizontal reference lines are the limiting constants $c=1$ and $c=2$ from Theorem~\ref{thm:bounded}. The finite-$d$ curves are already close to these targets, and the vertical distance between the two curves is close to the one-unit asymptotic gap throughout the displayed bounded-energy range. Section~\ref{sec:sim} gives the simulation design and Monte Carlo calculation.

\subsection{Critical Regime}
\label{sec:shellcritical}

The next regime asks how far this one-unit weak-signal advantage persists once the signal energy changes the quadratic statistic by the same order as its null fluctuation.

The intrinsic fluctuation scale of $T_d=\norm{X_d}^2$ is $\sqrt d$, because
\[
T_d=\norm{\varepsilon_d}^2+2\theta_d^\top\varepsilon_d+\norm{\theta_d}^2,
\quad
\varepsilon_d\sim N_d(0,I_d),
\]
and $\norm{\varepsilon_d}^2-d$ is of order $\sqrt d$. The critical-regime scaling $\norm{\theta_d}^2\asymp \sqrt d$ is therefore the first scaling at which signal energy changes $T_d$ by the same order as its stochastic fluctuations.

\begin{theorem}[Critical-regime risk limit]
\label{thm:main}
Suppose 
$\frac{\norm{\theta_d}^2}{\sqrt d}\to \beta\in[0,\infty)$.
Then
\[
\risk(\theta_d,\delta_{d,c})-\norm{\theta_d}^2\to L_c(\beta),
\]
where 
\begin{equation}\label{eq:L_func}
 L_c(\beta):=\E[h_c(Z_\beta)^2]-2\beta\,\E[h_c(Z_\beta)], 
\end{equation}
with $Z_\beta\sim N(\beta,2)$, $h_c(z):=\E_{q_{c,z}}[U]$, where $U\sim q_{c,z}$, and $q_{c,z}$ denoting the probability law on $(0,\infty)$ with density 
$$q_{c,z}(u)\propto u^{c/2-1}\exp\!\left(\frac{zu}{2}-\frac{u^2}{4}\right).$$
\end{theorem}

Through the exact risk identity of Lemma~\ref{lem:riskid}, the proof reduces to the analysis of the limit of the single term $\E[T_d\,s_{d,c}(T_d)^2]$, which is delivered by the following two ingredients. First, at the critical-rigeme scaling the signal displaces the quadratic statistic by the same order as its stochastic fluctuation, so the centered statistic $Z_d=(T_d-d)/\sqrt d$ converges to $N(\beta,2)$: the variance is the $\chi^2$ noise and the mean is the signal energy, while the cross term between signal and noise washes out. Second, because the statistic sits at a distance of order $\sqrt d$ from the origin, the latent variance $g$ is pinned near the boundary $g=0$ on the scale $1/\sqrt d$, and a local boundary-Laplace limit shows the rescaled shrinkage factor converges to the function $h_c$ in which the prior exponent $c$ survives. Combining the two, $T_d\,s_{d,c}(T_d)^2$ behaves like $h_c(Z_\beta)^2$, and centering by $\norm{\theta_d}^2$ leaves $L_c(\beta)$.

\begin{corollary}[Critical-regime risk gap and crossover]
\label{cor:deltatail}
If $\norm{\theta_d}^2/\sqrt d\to \beta\in[0,\infty)$, then
\[
\risk(\theta_d,\delta_{d,1})-\risk(\theta_d,\delta_{d,2})\to \Delta(\beta),
\quad
\Delta(\beta):=L_1(\beta)-L_2(\beta),
\]
and $\Delta(0)=-1$.
As $\beta\to\infty$,
\[
\Delta(\beta)=5\beta^{-2}+O(\beta^{-4}).
\]
Consequently, $\Delta(\beta)>0$ for all sufficiently large $\beta$. Since $\Delta(0)=-1$ and $\Delta$ is continuous, there exists at least one critical-regime crossover $\beta_*\in(0,\infty)$.
\end{corollary}

The critical regime is the transition regime in which the signal contribution to the quadratic statistic is comparable to the statistic's null fluctuation. The limiting risk comparison is therefore indexed by the single signal-strength parameter $\beta$: the SD-flat default is favored when the signal energy remains close to the origin ($\beta$ small), but the variance-flat harmonic benchmark can have smaller centered risk once $\beta$ is large enough. Thus the benchmark comparison is not a uniform dominance claim; it describes where each near-zero scale behavior is preferable in frequentist risk.

We numerically plot the finite-$d$ critical-regime risk gap $\risk(\theta_d,\delta_{d,1})-\risk(\theta_d,\delta_{d,2})$ together with the limiting curve $\Delta(\beta)$ in Figure~\ref{fig:shellcritical-gap-main}. Negative values favor the SD-scale benchmark and positive values favor the variance-flat benchmark. In the numerical calculation, the plotted asymptotic curve has a numerically identified zero at $\beta_{*}\approx 2.080$, where the centered risks of the two benchmarks are nearly tied. Section~\ref{sec:sim} gives the detailed design.

\begin{figure}[t!]
\caption{Critical-Regime Risk Gap}
\label{fig:shellcritical-gap-main}
\begin{center}
\includegraphics[width=0.78\textwidth,trim=446.7pt 0 0 0,clip]{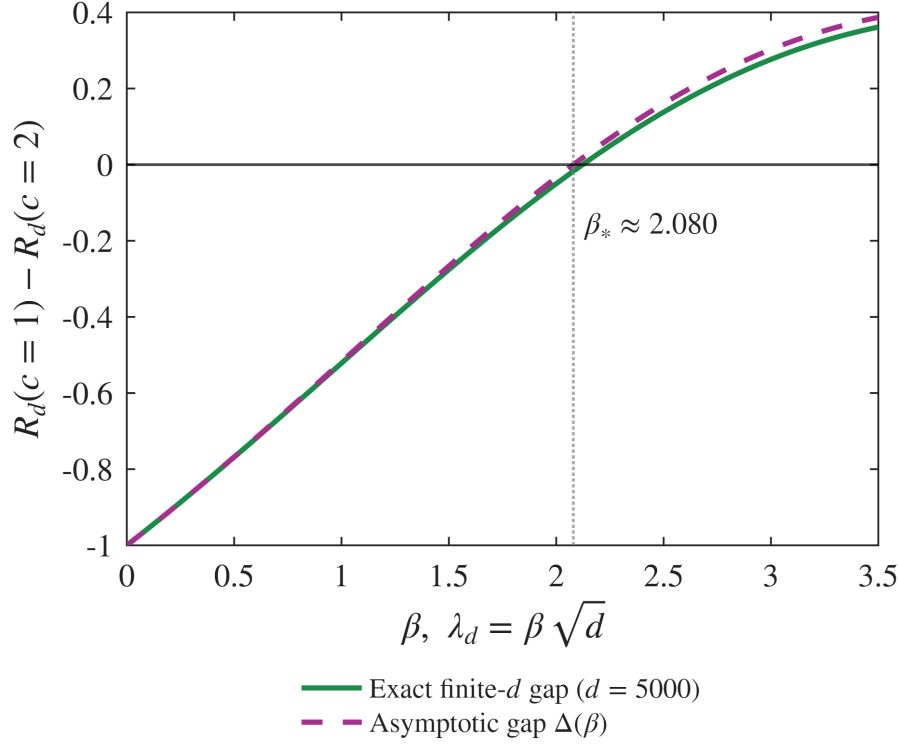}
\end{center}
{\footnotesize {\em Notes}: Critical-regime risk gap at \(d=5000\), with signal energy \(\lambda_d=\norm{\theta_d}^2=\beta\sqrt d\), together with the asymptotic limit \(\Delta(\beta)\), which has a numerically identified zero at \(\beta=\beta_{*}\approx 2.080\). Negative values favor the SD-scale benchmark; positive values favor the variance-flat benchmark.}
\setlength{\baselineskip}{4mm}
\end{figure}

The shape of $\Delta(\beta)$ is also informative. It starts at $\Delta(0)=-1$, so the weak-signal ordering survives for small signal energy in the critical regime. The upper-tail expansion $\Delta(\beta)=5\beta^{-2}+O(\beta^{-4})$ then forces the sign to become positive for sufficiently large signals in the critical regime.

\subsection{Strong Regime}
\label{sec:dense}

In the strong regime, $\norm{\theta_d}^2$ is of order $d$. The quadratic statistic $T_d=\norm{X_d}^2$ concentrates around $d+\norm{\theta_d}^2$, and the posterior for the global variance component moves away from the near-zero scaling that drives the critical-regime analysis.

In the following, Theorem~\ref{thm:strongfirst} gives the first-order limit $\risk(\theta_d,\delta_{d,c})/d\to \rho/(1+\rho)$ whenever $\norm{\theta_d}^2/d\to\rho>0$. Theorem~\ref{thm:dense2} refines this statement by identifying the $O(1)$ correction, and Corollary~\ref{cor:densegap} shows that any two fixed $c$ values are asymptotically indistinguishable at the $O(1)$ scale. 

\begin{theorem}[Strong-signal first-order universality]
\label{thm:strongfirst}
If 
$\frac{\norm{\theta_d}^2}{d}\to \rho>0$, 
then
\[
\frac{\risk(\theta_d,\delta_{d,c})}{d}\to \frac{\rho}{1+\rho}.
\]
In particular, the SD-flat and harmonic benchmarks are first-order asymptotically equivalent in the strong regime.
\end{theorem}

The main mechanism underlying Theorem \ref{thm:strongfirst} is that strong signals wash out the prior. When $\norm{\theta_d}^2\asymp d$, the normalized statistic $\tau_d=T_d/d$ obeys a law of large numbers and concentrates at the deterministic value $\tau_0=1+\rho$, so the stochastic fluctuations that drive the critical regime disappear. The exponent $c$ is therefore asymptotically irrelevant at this scale, which is the source of the universality. The random Bayes shrinkage factor accordingly collapses to the deterministic ridge coefficient $s_{d,c}(T_d)\to a(\tau_0)=1-1/\tau_0=\rho/(1+\rho)$, so $\delta_{d,c}$ behaves like constant shrinkage $a(\tau_0)X$. 

\begin{theorem}[Strong-signal second-order universality]
\label{thm:dense2}
If 
$\norm{\theta_d}^2=\rho d+\kappa+o(1)$ with 
$\rho>0$, then
\[
\risk(\theta_d,\delta_{d,c})
=
d\frac{\rho}{1+\rho}
+
\frac{\kappa}{(1+\rho)^2}
+
B(\rho)
+
o(1),
\quad
B(\rho):=\frac{2(1+2\rho+2\rho^2)}{(1+\rho)^3}.
\]
\end{theorem}

The refinement in Theorem \ref{thm:dense2} tracks two corrections that the first-order argument discards, and universality survives because they exactly cancel. We then have the following immediate corollary.

\begin{corollary}[Strong-signal indistinguishability]
\label{cor:densegap}
For any fixed $c_1,c_2>0$,
\[
\risk(\theta_d,\delta_{d,c_1})-\risk(\theta_d,\delta_{d,c_2})=o(1)
\]
whenever $\norm{\theta_d}^2=\rho d+\kappa+o(1)$ with $\rho>0$.
\end{corollary}

Figure~\ref{fig:dense-main} illustrates this result numerically, with the details of the numerical exercise available in  Section~\ref{sec:sim}.

\medskip

Taken together, the three regimes give a phase diagram for the default comparison. At bounded energy, the SD-flat prior has a fixed one-unit advantage. Under the critical-regime scaling, that advantage holds only for sufficiently small signal energy and crosses over at a finite threshold on the signal-energy growth rate. At energy proportional to $d$, the posterior is no longer governed by the zero-scale boundary, and the two fixed-exponent benchmarks share the same asymptotic risk through the second-order risk expansion. 

Overall, this three-regime transition gives a geometric risk map for the SD-flat versus variance-flat choice. The SD-flat prior is preferable in the weak regime and for sufficiently small critical-regime signal strength, while the variance-flat prior can have smaller centered risk for larger critical-regime signals. Once signal energy is proportional to $d$, however, the zero-scale geometry no longer affects risk through the second-order expansion.

\section{Proper Defaults and Common Priors}
\label{sec:proper}

Section~\ref{sec:model} compares the two benchmark geometries inside the radial power family. We now transfer those conclusions to proper priors, to common single global-scale hyperpriors, and to bounded coordinate-multiplier priors. We then use the same near-zero scale calculus to interpret common priors beyond a single global scale that fall outside the bounded-multiplier theorem. The key quantity is the near-zero exponent of the relevant standard-deviation scale.

\subsection{Robustness to Proper Default Priors}
\label{subsec:properdefaults}
Fix $c>0$. Let $\ell:[0,\infty)\to[0,\infty)$ be a nonincreasing function that is continuous at $0$ with $\ell(0)\in (0,\infty)$ and
\[
\int_0^\infty g^{c/2-1}\ell(g)\,dg<\infty.
\]
Based on $\ell$, we may define the following \emph{proper} hierarchical prior
\begin{equation}\label{eq:proper}
 \theta\mid g\sim N_d(0,gI_d),
\quad
p_\ell(g)\propto g^{c/2-1}\ell(g), 
\end{equation}
and let $\delta_{d,c,\ell}$ be the corresponding posterior mean.\footnote{Equivalently, on the SD scale, if $\theta\mid\tau\sim N_d(0,\tau^2I_d)$ and $p_m(\tau)\propto \tau^{c-1}m(\tau)$, where $m$ is bounded, nonincreasing, continuous and positive at zero, and $\int_0^\infty \tau^{c-1}m(\tau)\,d\tau<\infty$, then the same transfer conclusions apply.}

\begin{proposition}[Robustness to monotone properizations] 
\label{prop:proper} 
Under the proper hierarchical prior \eqref{eq:proper}, the conclusions of Theorems~\ref{thm:bounded} and~\ref{thm:main} continue to apply under their respective signal regimes to the posterior mean $\delta_{d,c,\ell}$ with the same exponent $c$.
\end{proposition}

The properizing factor $\ell$ does not affect the limiting low-signal risks because those limits are determined at the variance-scale endpoint $g=0$. In the weak and critical regimes, $T_d=d+O_{\Pbb}(\sqrt d)$, so the posterior contribution to the shrinkage factor comes from $g=O(d^{-1/2})$. On that scale, $\ell(g)=\ell(0)+o(1)$, and the constant $\ell(0)$ cancels from the numerator and denominator of the posterior mean. Thus the local power $g^{c/2-1}$ fixes the local near-zero asymptotic problem, while monotonicity and integrability of $\ell$ only keep the properized prior from contributing mass at larger, nonlocal scales.

Proposition~\ref{prop:proper} covers several proper defaults used in practice, including the half-$t$ family advocated by \citet{Gelman2006} and the half-Cauchy special case emphasized by \citet{PolsonScott2012}:
\[
p(\tau)\propto \mathbf 1_{[0,A]}(\tau)
\quad\Longleftrightarrow\quad
p(g)\propto g^{-1/2}\mathbf 1_{[0,A^2]}(g)
\]
(truncated flat SD),
\[
p(g)\propto g^{-1/2}\Bigl(1+\frac{g}{\nu_0 s^2}\Bigr)^{-(\nu_0+1)/2}
\]
(half-$t$ on SD, including half-Cauchy when $\nu_0=1$), and
\[
p(g)\propto \mathbf 1_{[0,A]}(g)
\]
(truncated flat variance). Thus the one-unit weak-signal advantage of flat on SD over flat on variance is not an artifact of the improper power family; it persists for the corresponding proper defaults as well.

The SD-scale formulation in Proposition~\ref{prop:proper} turns the benchmark radial index $c$ into a classification rule on the standard-deviation scale. For a single global-scale hyperprior, the weak-signal and critical-regime limits depend only on the near-zero exponent of the global SD density.

\subsection{Common Single Global-Scale Hyperpriors}
\label{subsec:globalpriors}

This subsection classifies commonly used single global-scale hyperpriors according to whether their near-zero standard-deviation density has the SD-flat geometry ($c=1$), the variance-flat geometry ($c=2$), or a more general power $c$.
The SD-scale formulation in Proposition~\ref{prop:proper} gives immediate consequences for several priors that are widely used in high-dimensional shrinkage work. Consider the single global-scale Gaussian hierarchy
\[
\theta\mid \tau\sim N_d(0,\tau^2 I_d),
\]
and let $\delta_{d,p}$ denote the posterior mean under hyperprior $p(\tau)$.

\paragraph{Flat-SD, half-\texorpdfstring{$t$}{t}, half-Cauchy, half-normal, exponential, and truncated-flat priors.}
If $p(\tau)$ is half-$t_{\nu_0}(s)$, half-Cauchy$(s)$, half-normal$(s)$, exponential$(b)$, or truncated flat on $[0,A]$, then the density on the SD scale is positive and finite at zero, so Theorems~\ref{thm:bounded} and~\ref{thm:main} continue to apply to $\delta_{d,p}$ with $c=1$.

\paragraph{\texorpdfstring{$R^2$}{R2}-based priors.}
Suppose $W=w_0\tau^2$ for some fixed $w_0>0$ and
\[
R^2=\frac{W}{1+W}\sim \mathrm{Beta}(a,b),
\quad a,b>0.
\]
Then
\[
p(\tau)=\frac{2w_0^a}{\mathrm{B}(a,b)}\tau^{2a-1}(1+w_0\tau^2)^{-(a+b)},
\quad \tau>0,
\]
and therefore the near-zero SD-scale exponent is $c=2a$, so Theorems~\ref{thm:bounded} and~\ref{thm:main} continue to apply to $\delta_{d,p}$ with $c=2a$.
In particular, $a=\tfrac12$ gives the SD-scale class, whereas $a=1$ gives the variance-flat class; for any $a_1<a_2$, Theorem~\ref{thm:bounded} gives the weak-signal gap $2(a_1-a_2)$.

\paragraph{Gamma priors on variance.}
If $g=\tau^2\sim \mathrm{Gamma}(a,b)$, then
\[
p(\tau)=\frac{2b^a}{\Gamma(a)}\tau^{2a-1}e^{-b\tau^2},
\quad \tau>0,
\]
the near-zero SD-scale exponent is again $c=2a$, so gamma variance priors fall under Theorems~\ref{thm:bounded} and~\ref{thm:main} with $c=2a$. In particular, exponential priors on $g$ ($a=1$) belong to the same variance-flat class as the improper flat prior on $g$.

Therefore, by Proposition~\ref{prop:proper}, the Section~\ref{sec:model} results translate directly: SD-scale priors inherit the $c=1$ weak-signal and critical-regime limits, variance-scale priors inherit the $c=2$ limits, and the more general beta or gamma powers use their corresponding exponent $c$.

\paragraph{Scale calibration.}
This exponent classification concerns the local shape of $p(\tau)$ at zero; it does not choose the numerical scale. For the horseshoe in regression, \citet{PiironenVehtari2017Hyperprior,PiironenVehtari2017Sparsity} propose calibrating that scale through prior information about sparsity. In their Gaussian-regression normalization, if $p_0$ is the prior guess for the number of relevant predictors among $D$ predictors, then the natural global-scale level is
\[
\tau_0=\frac{p_0}{D-p_0}\frac{\sigma}{\sqrt n}.
\]
This separates two components of the prior specification: the weak-regime comparison favors the $c=1$ SD-scale shape over the $c=2$ variance-flat benchmark, while the width of that prior can be chosen from substantive sparsity information rather than from a universal unit-scale convention.

\subsection{Beyond a Single Global Scale}
\label{subsec:beyondglobal}

Many high-dimensional priors add coefficient-specific latent structure on top of a common scale. This subsection studies a precise transfer question: after such a component of the hierarchy is added, which part of the Section~\ref{sec:model} comparison still comes from the common scale, and which part can be changed by the additional component?

We separate the problem into two cases. For bounded coordinate-multiplier priors, an exact risk-transfer theorem shows that the same near-zero SD-scale exponent $c$ determines the weak-signal and critical-regime limits, just as in the single-global benchmark. For common sparse or heavy-tailed priors outside this bounded class, the exponent still classifies the common global-scale component, but it does not by itself yield a complete risk theorem: local-scale tails, model-size priors, or dimension-dependent allocation priors can also enter the leading risk calculation.

We use the standard names for the resulting architectures. A \emph{global--local prior} is a continuous shrinkage prior with a common scale and positive local scales, typically
\[
\theta_j\mid \tau,\lambda_j\sim N(0,\tau^2\lambda_j^2),
\quad \lambda_j>0.
\]
A \emph{spike-and-slab prior} is a model-selection prior with inclusion indicators and, in the point-mass case, exact prior mass on coordinate subspaces. Dirichlet--Laplace and R2-D2 priors \citep{BhattacharyaPatiPillaiDunson2015DL,ZhangNaughtonBondellReich2022R2D2} add a different type of structure: a common total scale or model-fit parameter is distributed across coordinates through a Dirichlet weight vector. These architectures all contain a common scale, but the additional component has a different mathematical role in each case.

\paragraph{Exact transfer for bounded coordinate multipliers.}

The formal transfer theorem uses the following coordinate-multiplier envelope. It keeps a single common variance scale $g=\tau^2$, but allows each coordinate to receive a bounded multiplier $A_j$.

Consider the coordinate-multiplier hierarchy
\begin{equation}
\label{eq:coord-mult-model}
X_j\mid \theta_j\sim N(\theta_j,1),
\quad
\theta_j\mid g,A_j\sim N(0,gA_j),
\quad A_j\overset{\mathrm{iid}}{\sim}F,
\end{equation}
with common variance prior
\[
\pi(g)\propto g^{c/2-1}\ell(g),
\quad g=\tau^2.
\]
When $A_j=\lambda_j^2$ with positive continuous local scales, \eqref{eq:coord-mult-model} is a bounded global--local normal scale mixture. When $A_j=Z_j\in\{0,1\}$ is Bernoulli with fixed success probability $q$, it gives the fixed-$q$ point-mass spike-and-slab prior
\[
Z_j\overset{\mathrm{iid}}{\sim}\operatorname{Bernoulli}(q),
\quad
\theta_j\mid g,Z_j\sim (1-Z_j)\delta_0+Z_jN(0,g).
\]
Thus the coordinate-multiplier model is a common mathematical representation; the subclasses retain their usual statistical interpretations.

The posterior mean under \eqref{eq:coord-mult-model} is generally non-radial, because different coordinates can be shrunk by different posterior multiplier information. Proposition~\ref{prop:bounded-mult-transfer} shows that this non-radiality does not change the low-signal local near-zero limit as long as the multipliers are bounded and have positive mean. The analytic posterior-mean identity and the coordinate marginal used in the proof are given in Appendix~\ref{app:coord-mult-proof}.

\begin{proposition}[Bounded coordinate-multiplier transfer]
\label{prop:bounded-mult-transfer}
Assume $0\le A\le A_+<\infty$ almost surely and $\mu:=\E A\in(0,\infty)$. Let
\[
\pi(g)\propto g^{c/2-1}\ell(g)\mathbf 1_{(0,G)}(g),
\]
where $G<\infty$ and $\ell$ is bounded, nonnegative, nonincreasing, continuous at zero, and satisfies $\ell(0)>0$. Let $\delta_d$ be the full posterior mean under \eqref{eq:coord-mult-model}. If
\[
\sup_{d,j}|\theta_{dj}|<\infty,
\]
then, the conclusions of Theorems~\ref{thm:bounded} and~\ref{thm:main} continue to apply under their respective signal regimes to $\delta_d$ with the same exponent $c$.
\end{proposition}

The proof also explains why the multiplier distribution does not affect the limit. The multiplier distribution enters the local near-zero likelihood through its mean $\mu=\E A$. After the change of variables $u=\mu\sqrt d\,g$, this factor disappears from the limiting likelihood, leaving the same local near-zero asymptotic problem as in Section~\ref{sec:model}. Hence, within the bounded coordinate-multiplier class, an SD-scale common prior retains the one-unit weak-signal advantage over its variance-scale analogue.

\paragraph{What does not transfer as an exact theorem.}

Proposition~\ref{prop:bounded-mult-transfer} is not a universal theorem for every prior that contains a global scale. Outside the bounded-multiplier envelope, the common-scale exponent still says whether the common scale is SD-like or variance-like near zero, but leading risk can also be affected by other prior components: local-scale tails, model-size priors, or dimension-dependent allocation priors. Hence the constants $L_c(\beta)$ should not be read as applying automatically to horseshoe-type, sparse spike-and-slab, Dirichlet--Laplace, or R2-D2 priors. Appendix~\ref{app:beyond-details} gives the corresponding calculations. The simulations in Section~\ref{sec:sim} provide qualitative support outside the theorem: in the many-weak designs considered there, the $c=1$ common-scale law has lower weak-signal excess risk than its $c=2$ analogue within both the global--local and spike-and-slab architectures.

\section{Numerical Illustrations}
\label{sec:sim}

This section gives finite-dimensional illustrations of the risk regimes in Section~\ref{sec:model}, the transfer to representative proper single global-scale hyperpriors, and the weak-signal comparison within two representative architectures beyond a single global scale. For the benchmark family and for single global-scale hyperpriors, the posterior mean is radial, so the displayed risks are exact finite-$d$ evaluations up to numerical integration error. The final numerical comparison is different: the posterior mean is non-radial, and the reported risks combine Monte Carlo evaluation with deterministic quadrature over the common global or slab scale.

To represent many weak effects, the benchmark and single-global-scale calculations use the equal-effects vector
\[
\theta_d^{\mathrm{eq}}(\lambda_d)
=
\sqrt{\lambda_d/d}\,\mathbf 1_d,
\quad
\lambda_d=\norm{\theta_d}^2.
\]
For the radial rules studied in the first three numerical exercises, the risk depends on $\theta_d$ only through $\lambda_d$, so this representation is interpretive rather than restrictive. In the final comparison beyond a single global scale, the displayed many-weak vector is part of the simulation design.

\paragraph{Default priors.}

Figures~\ref{fig:weak-benchmark-main} and~\ref{fig:shellcritical-gap-main} give the benchmark checks next to the corresponding theoretical statements. Figure~\ref{fig:weak-benchmark-main} verifies Theorem~\ref{thm:bounded}: already at $d=2000$, the excess risk $\risk(\theta_d,\delta_{d,c})-\lambda_d$ is close to its limiting constant $c$. In Table~\ref{tab:regime-check-main}, for example, the $c=1$ and $c=2$ risks are $(1.981,2.971)$ at $\lambda=1$ and $(4.922,5.882)$ at $\lambda=4$, so the SD-scale benchmark retains an approximately one-unit advantage throughout the displayed weak-signal range. Figure~\ref{fig:shellcritical-gap-main} visualizes Theorem~\ref{thm:main} and Corollary~\ref{cor:deltatail}: the finite-$d$ critical-regime gap $\risk(\theta_d,\delta_{d,1})-\risk(\theta_d,\delta_{d,2})$ tracks the limit $\Delta(\beta)$ closely, is zero at $\beta=\beta_{*}\approx 2.080$, and reverses sign by $\beta=3$ ($-6.952$ versus $-7.229$).

Figure~\ref{fig:dense-main} verifies Theorem~\ref{thm:dense2}. At both $d=200$ and $d=1000$, the scaled risks $\risk(\theta_d,\delta_{d,c})/d$ for $c=1$ and $c=2$ are nearly indistinguishable, while the second-order approximation materially sharpens the first-order limit when $d$ is moderate. Table~\ref{tab:regime-check-main} records the same collapse at $d=200$: the two scaled risks agree to the displayed precision at $\rho=1$ and $\rho=4$, and differ by only $0.002$ at $\rho=0.25$. The calculations illustrate Corollary~\ref{cor:densegap}: once the signal energy is of order $d$, the distinction between flat on variance and flat on SD disappears at the $O(1)$ scale.

\begin{figure}[t!]
\caption{Strong-signal universality}
\label{fig:dense-main}
\begin{center}
\includegraphics[width=0.95\textwidth]{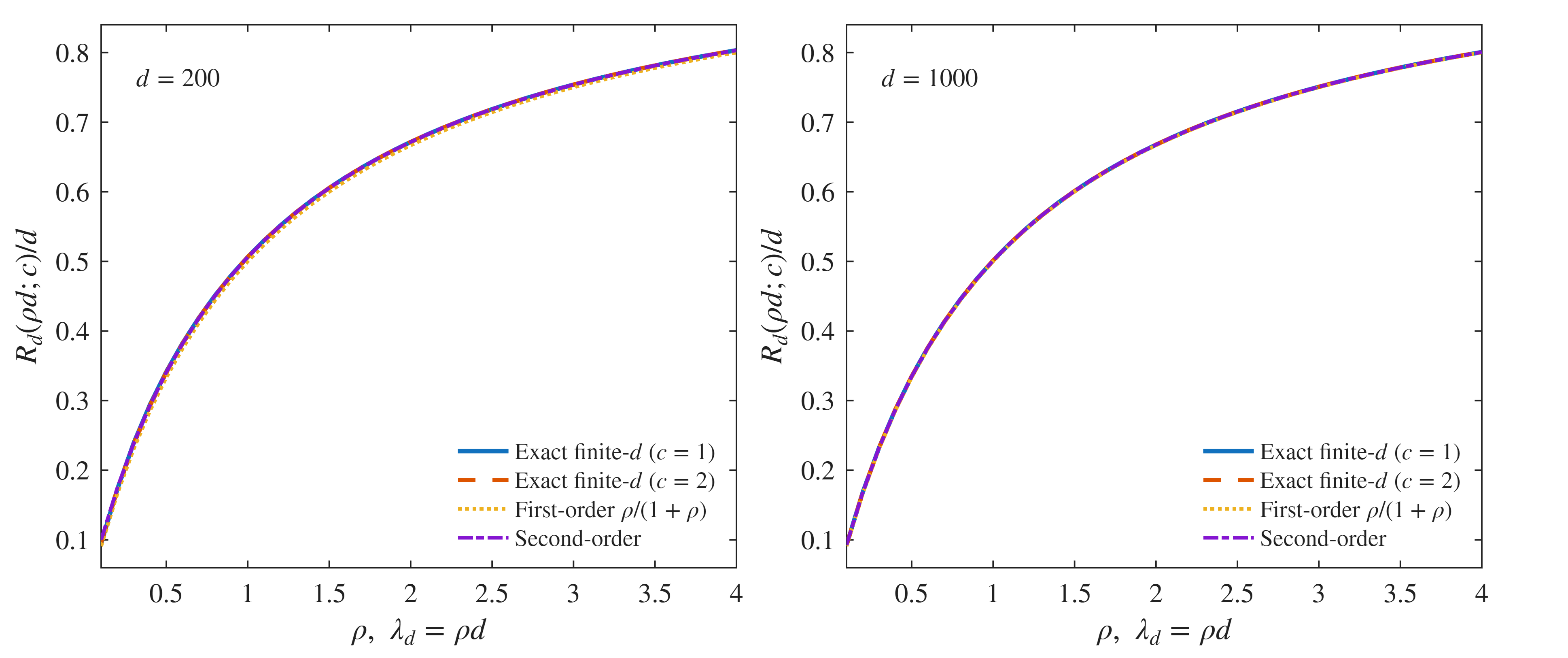}
\end{center}
{\footnotesize {\em Notes}: The exact finite-$d$ scaled risks for $c=1$ and $c=2$ are nearly identical, and the second-order approximation improves substantially on the first-order limit at moderate dimensions.}
\setlength{\baselineskip}{4mm}
\end{figure}

Table~\ref{tab:regime-check-main} gives selected finite-$d$ values from the weak-signal, critical-regime, and strong-signal calibrations. At the null, the risks are exactly $1.000$ and $2.000$ for $c=1$ and $c=2$. At the displayed critical-regime calibration $\beta_{*}\approx2.080$ the two centered risks are nearly tied, while at $\beta=3$ the ordering has reversed. In the strong-signal rows, the scaled risks for $c=1$ and $c=2$ coincide to the displayed precision except for the smallest $\rho$ value.

\input{sim_tables/tab_regime_check_main_20260507.tex}

\paragraph{Proper hyperpriors.}

Figure~\ref{fig:transfer-main} and Table~\ref{tab:transfer-priors-main} verify the transfer results from Section~\ref{sec:proper}, namely Proposition~\ref{prop:proper} and the single-global-scale dictionary in Section~\ref{subsec:globalpriors}. The $c=1$ class---half-Cauchy, half-normal, truncated flat on $\tau$, and Beta-$R^2$ with $a=\tfrac12$---clusters around the SD-flat benchmark: the displayed null risks range from $0.965$ to $1.000$, and the critical-regime values at $\beta=1$ range from $-0.302$ to $-0.272$. The $c=2$ class---Gamma on $g$ with $a=1$, truncated flat on $g$, and Beta-$R^2$ with $a=1$---clusters around the variance-flat benchmark, with null risks from $1.929$ to $2.000$ and $\beta=1$ values from $0.184$ to $0.250$. Finite-$d$ differences within each class remain, as expected for proper priors, but the classwise ordering is stable and becomes especially sharp away from the critical-regime crossover.

\begin{figure}[t!]
\caption{Transfer to representative proper single global-scale hyperpriors at $d=5000$}
\label{fig:transfer-main}
\begin{center}
\includegraphics[width=0.95\textwidth]{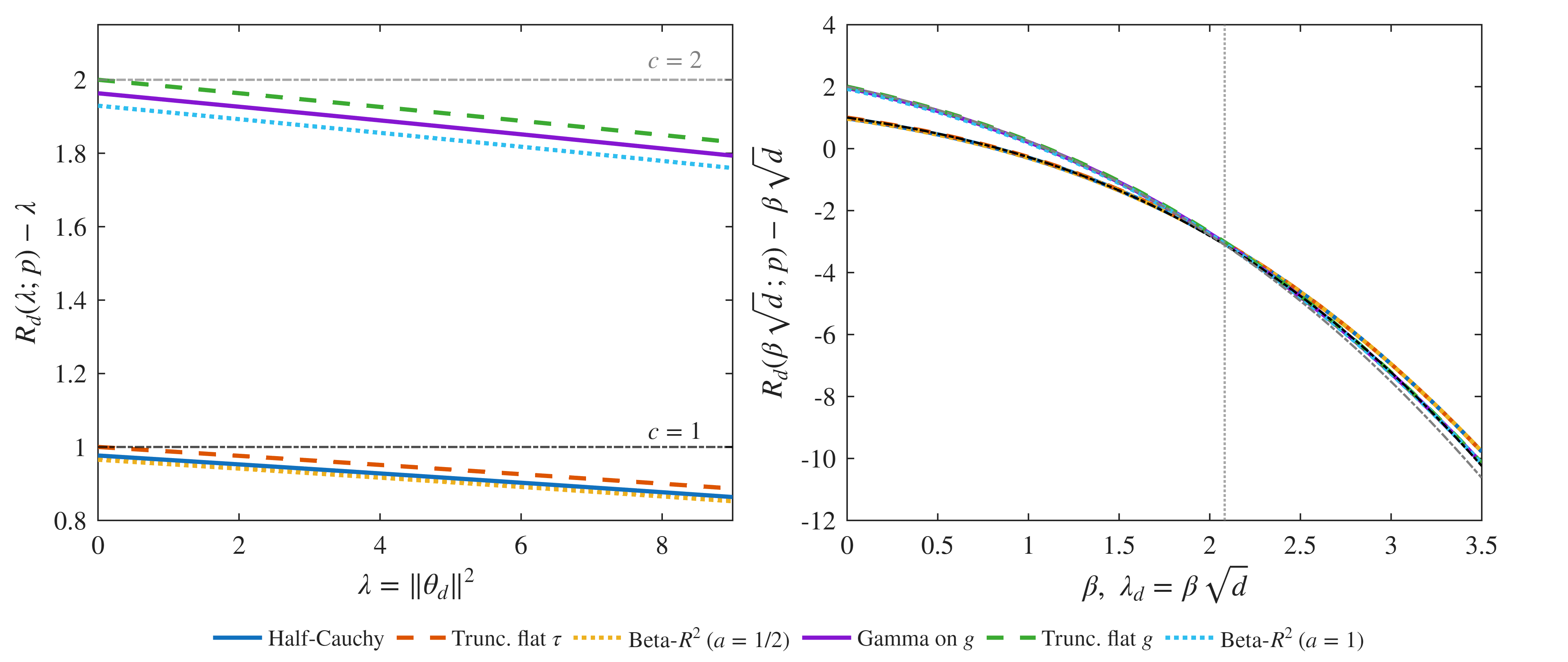}
\end{center}
{\footnotesize {\em Notes}: Left: weak-regime excess risk. Right: critical-regime centered risk. The priors cluster by the near-zero exponent of the SD-scale density, which is the numerical content of the transfer principle.}
\setlength{\baselineskip}{4mm}
\end{figure}

\input{sim_tables/tab_transfer_priors_main.tex}

\paragraph{Beyond a single global scale.}

Section~\ref{subsec:beyondglobal} separates exact transfer results from architecture-specific risk behavior. The numerical exercise here focuses on a narrower comparison: within a fixed architecture, only the common-scale law is changed from the $c=1$ SD-scale class to the $c=2$ variance-scale class.

We consider the many-weak design
\[
m_d=\lfloor d^{3/4}\rfloor,
\quad
\theta_{d,j}^{\mathrm{mw}}(\lambda)=(-1)^{j-1}\sqrt{\lambda/m_d}\,\mathbf 1\{j\le m_d\},
\quad j=1,\ldots,d,
\]
so that $\norm{\theta_d^{\mathrm{mw}}(\lambda)}^2=\lambda$, the number of active coordinates diverges, and each active coordinate remains weak. The simulations use $d=500$ and $\lambda\in\{0,1,2,4\}$.

The first architecture is a horseshoe-type global--local normal scale mixture with half-Cauchy local scales,
\[
\theta_j\mid \lambda_j,\tau\sim N(0,\tau^2\lambda_j^2),
\quad
\lambda_j\sim \mathrm{half\mbox{-}Cauchy}.
\]
Its $c=1$ version assigns a half-Cauchy prior to the common SD scale $\tau$, whereas its $c=2$ version assigns an exponential prior to the common variance $g=\tau^2$. The second architecture is the standard Bernoulli spike-and-slab model,
\[
\theta_j\mid z_j,\gamma\sim (1-z_j)\delta_0+z_jN(0,\gamma^2),
\quad
z_j\overset{\mathrm{iid}}{\sim}\operatorname{Bernoulli}(q_d).
\]
For the simulation we set $q_d=m_d/d$, so that the prior expected model size $dq_d$ equals the number of active coordinates in the design. This oracle calibration is used only to isolate the effect of the common slab-scale law; in data analysis $q$ is typically assigned a hyperprior or estimated by marginal likelihood. The $c=1$ version places a half-Cauchy prior on the common slab SD $\gamma$, and the $c=2$ version replaces that component by an exponential prior on the slab variance $\gamma^2$.

Figure~\ref{fig:beyond-global-twoarch} and Table~\ref{tab:beyond-global-twoarch} show that the within-architecture comparison follows the same direction in both examples. Replacing the $c=1$ common-scale law by the corresponding $c=2$ law increases weak-signal excess risk at every displayed value of $\lambda$. For the global--local architecture, the increase ranges from $0.86$ to $0.94$; for the spike-and-slab architecture, it ranges from $0.64$ to $0.72$. At $\lambda=4$, for instance, the excess risks are $0.97$ and $1.83$ in the global--local architecture, and $0.62$ and $1.26$ in the spike-and-slab architecture. The risk curves remain architecture-specific. Nevertheless, after the architecture is fixed, the common-scale exponent continues to organize the weak-signal comparison in the same direction as in the single-global benchmark.

\begin{figure}[t!]
\caption{Weak-signal many-weak comparison for two representative architectures beyond a single global scale at $d=500$}
\label{fig:beyond-global-twoarch}
\begin{center}
\includegraphics[width=0.95\textwidth]{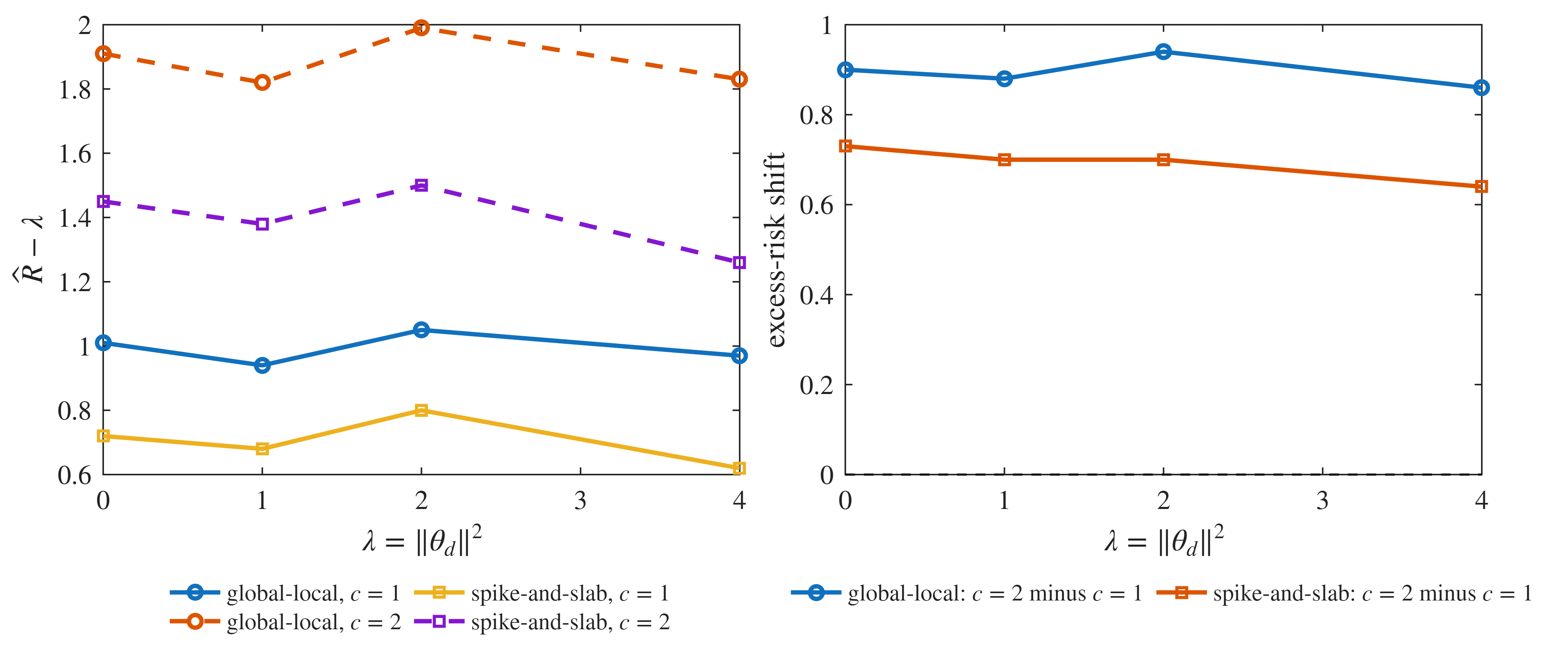}
\end{center}
{\footnotesize {\em Notes}: Left: excess-risk curves for the $c=1$ and $c=2$ variants of a global--local prior and a spike-and-slab prior. Right: within-architecture increase in excess risk from replacing the $c=1$ common-scale law by its $c=2$ analogue.}
\setlength{\baselineskip}{4mm}
\end{figure}

\input{sim_tables/tab_beyond_global_twoarch_20260507.tex}

Overall, the numerical evidence reinforces the main conclusions. The weak-signal one-unit advantage is already visible at moderate dimension, the critical regime exhibits the predicted crossover, strong-signal sequences erase the benchmark distinction up to second order, representative proper single global-scale hyperpriors inherit the same low-signal behavior through their near-zero SD-scale exponent, and the two representative architectures beyond a single global scale display the same within-architecture $c=1$ versus $c=2$ ordering in many-weak signals.

\section{Conclusion}
\label{sec:discussion}

This paper isolates the near-zero geometry of a global Gaussian shrinkage scale as the relevant object in high-dimensional low-signal risk. For the two canonical defaults, flat on the standard deviation improves on flat on the variance by one asymptotic risk unit in the weak regime, the critical regime exhibits a signal-strength crossover, and strong-signal sequences erase the distinction up to second order.

The transfer results show how this benchmark comparison propagates to the priors used in practice. For a single global-scale hyperprior, the near-zero exponent of the global SD density is the only statistic that matters for weak-signal and critical-regime isotropic risk. For bounded coordinate-multiplier normal mixtures, the same limit survives even though the posterior mean is non-radial. This bounded class contains bounded global--local normal scale mixtures and Bernoulli spike-and-slab priors with fixed inclusion probability as distinct subclasses. For commonly used priors outside this class, the classification identifies the common global-scale component and the additional component that enters the asymptotic problem: finite-mean local mixtures have the same first-order small-scale geometry, horseshoe-type local-scale distributions create a heavy-tail near-zero behavior in which the effective small-scale parameter is $\tau$, and sparse spike-and-slab priors are governed partly by the model-size prior.

Several questions remain open. A sharper characterization of the critical-regime crossover set, including uniqueness of the zero of $\Delta(\beta)$, would refine the benchmark phase diagram. On the prior side, the next theoretical steps are to extend the bounded coordinate-multiplier theorem to unbounded finite-moment multipliers, to develop the heavy-tail near-zero theory needed for horseshoe-type local-scale distributions, and to analyze triangular model-size priors, such as sparse spike-and-slab, and dimension-dependent Dirichlet-weight priors, such as Dirichlet--Laplace and R2-D2. These extensions would separate the contribution of the common global scale from the additional local, Dirichlet-weight, or model-size prior that determines risk outside the bounded class.

~

\noindent \textbf{Disclosure of Generative AI Usage}: We used generative AI tools (ChatGPT 5.5, Claude Opus 4.8, Gemini 3.5, and Refine.ink) for assistance with coding, drafting/editing, and critical reviews. We assume full responsibility for the final manuscript.

\bibliographystyle{ecta}
\bibliography{GeoBayesRef}

\newpage
\appendix
\begin{center}
\textbf{\huge Appendix}
\end{center}

\section{Proofs}
\label{app:proofs}

\subsection{Results in Section~\ref{sec:setup}}
\label{app:modelproofs}

\begin{proof}[\textbf{Proof of Proposition~\ref{prop:mixture}}]
Up to constants,
\[
N_d(\theta;0,gI_d)\, g^{c/2-1}
\propto
g^{-d/2} \exp\!\left(-\frac{\norm{\theta}^2}{2g}\right) g^{c/2-1}
=
g^{-(d-c+2)/2}\exp\!\left(-\frac{\norm{\theta}^2}{2g}\right).
\]

Hence, for fixed $\theta\neq 0$,
\[
\int_0^\infty g^{-(d-c+2)/2}\exp\!\left(-\frac{\norm{\theta}^2}{2g}\right)\,dg
=
\int_0^\infty g^{-\alpha}e^{-\beta/g}\,dg,
\quad
\alpha:=\frac{d-c+2}{2},\ \beta:=\frac{\norm{\theta}^2}{2}.
\]
Set $u=\beta/g$ (equivalently $g=\beta/u$, $dg=-\beta u^{-2}du$). Then
\[
\int_0^\infty g^{-\alpha}e^{-\beta/g}\,dg
=\beta^{1-\alpha}\int_0^\infty u^{\alpha-2}e^{-u}\,du
=\beta^{1-\alpha}\Gamma(\alpha-1)
=2^{(d-c)/2}\Gamma\!\left(\frac{d-c}{2}\right)\norm{\theta}^{-(d-c)}.
\]
Thus the marginal prior kernel in $\theta$ is proportional to $\norm{\theta}^{-(d-c)}$.

\end{proof}

\begin{proof}[\textbf{Proof of Lemma~\ref{prop:beta}}]
Starting from Proposition~\ref{prop:mixture}, the posterior density of $g$ given $T_d=t$ is
proportional to
\[
g^{c/2-1}(1+g)^{-d/2}\exp\!\left(-\frac{t}{2(1+g)}\right), \quad g>0.
\]
Applying the change of variable $y=g/(1+g)$ with inverse $g=y/(1-y)$ and Jacobian $(1-y)^{-2}$, we have
\[
\exp\!\left(-\frac{t}{2(1+g)}\right)=\exp\!\left(-\frac{t}{2}(1-y)\right)=e^{-t/2}\,e^{ty/2}.
\]
Noting $p_Y(y)=p_g(g(y))\,|dg/dy|$ and collecting factors,
\[
p_Y(y)
\;\propto\;
\left(\frac{y}{1-y}\right)^{c/2-1}
\big((1-y)^{-1}\big)^{-d/2}
e^{-t/2}\,e^{ty/2}\,
(1-y)^{-2}.
\]
The power of $y$ is $c/2-1$, while the exponents of $(1-y)$ sum to
\[
-(c/2-1)+d/2-2=\frac{d-c}{2}-1.
\]
Absorbing the constant $e^{-t/2}$ into the normalizing constant yields the displayed density
\[
p_Y(y)\;\propto\; y^{c/2-1}(1-y)^{(d-c)/2-1}e^{ty/2}, \quad 0<y<1.
\]

\end{proof}

The remaining setup identity is Lemma~\ref{lem:riskid}. It is proved from the posterior moment identity in the proposition below, which is also used later.

Let
\[
\mu_d(t):=\E[Y\mid T_d=t]=s_{d,c}(t),
\quad
q_d(t):=\E[Y^2\mid T_d=t].
\]

\begin{proposition}[Posterior moment identity]
\label{prop:moment}
For every $d>c$ and every $t\ge 0$,
\[
t\,q_d(t)=c+(t-d)\mu_d(t).
\]
Consequently, there exists $C_c<\infty$ such that
\[
t\,s_{d,c}(t)^2 \le C_c\left(1+\left(\frac{(t-d)_+}{\sqrt d}\right)^2\right)
\]
for all $d$ and all $t\ge 0$.
\end{proposition}

\begin{proof}[\textbf{Proof of Proposition~\ref{prop:moment}}]
By Lemma~\ref{prop:beta}, the posterior density of $Y$ given $T_d=t$ is proportional to
\[
h(y):=y^{c/2}(1-y)^{(d-c)/2}e^{ty/2}=y(1-y)\,f(y).
\]
Taking derivative of $h$ yields
\[
h'(y)=f(y)\left[\frac{c}{2}(1-y)-\frac{d-c}{2}\,y+\frac{t}{2}\,y(1-y)\right]
=f(y)\left[\frac{c}{2}+\frac{t-d}{2}\,y-\frac{t}{2}\,y^2\right].
\]
Because $c>0$ and $d>c$, both exponents $c/2$ and $(d-c)/2$ in $h$ are positive, so the boundary
terms vanish: $h(0)=h(1)=0$. Hence $\int_0^1 h'(y)\,dy=h(1)-h(0)=0$, that is,
\[
0=\int_0^1 f(y)\left[\frac{c}{2}+\frac{t-d}{2}\,y-\frac{t}{2}\,y^2\right]dy
=\frac{c}{2}\,\kappa+\frac{t-d}{2}\int_0^1 y f(y)\,dy-\frac{t}{2}\int_0^1 y^2 f(y)\,dy.
\]
where $\kappa:=\int_0^1 f(y)\,dy$. Dividing by $\kappa/2$ and using the definitions of $\mu_d(t)$ and $q_d(t)$ gives
$0=c+(t-d)\mu_d(t)-t\,q_d(t)$, i.e.
\[
t\,q_d(t)=c+(t-d)\mu_d(t).
\]
Since $q_d(t)\ge \mu_d(t)^2$, we have 

\begin{equation}
\label{eq:moment-ineq}
t\,\mu_d(t)^2 \le c+(t-d)\mu_d(t)
\end{equation}
We write $\mu$ in short for $\mu_d(t)$, and note that $0\le\mu\le 1$ because $Y\in(0,1)$. 

We then consider two separate cases:

\emph{Case 1: $t<d/2$.} Here $t-d<0$, so $(t-d)\mu\le 0$ and \eqref{eq:moment-ineq} gives directly
$t\,\mu^2\le c$.

\emph{Case 2: $t\ge d/2$.} Set
$z:=\frac{t-d}{\sqrt d}$, and $v:=\sqrt d\,\mu$, so that $\mu=v/\sqrt d$ and $t-d=\sqrt d\,z$. Then $t\,\mu^2=(t/d)v^2$ and $(t-d)\mu=zv$, so
\eqref{eq:moment-ineq} becomes
\begin{equation}
\label{eq:vz}
\frac{t}{d}\,v^2 \le c+zv.
\end{equation}
Since $t\ge d/2$ we have $t/d\ge 1/2$, and as $v\ge 0$ this yields $\tfrac12 v^2\le (t/d)v^2\le c+zv$,
i.e.\ $v^2-2zv-2c\le 0$. Completing the square and solving the inequality in $v$, we have
\[
v \le z+\sqrt{z^2+2c} \le z+|z|+\sqrt{2c} = 2z_+ + \sqrt{2c},
\quad \text{with } z_+:=\max\{z,0\}
\]
Substituting back into \eqref{eq:vz} and
using $zv\le z_+v$, we obtain
\[
t\,\mu^2=\frac{t}{d}\,v^2 \le c+zv \le c+z_+v \le c+z_+\bigl(2z_+ + \sqrt{2c}\bigr)
= 2z_+^2+\sqrt{2c}\,z_+ + c.
\]
With the elementary bound $z_+\le\tfrac12(1+z_+^2)$ (equivalently $(1-z_+)^2\ge0$), the right-hand
side is at most $C_c(1+z_+^2)$ with $C_c:=2+\sqrt{2c}+c$. Since $z_+=(t-d)_+/\sqrt d$ and the case
$t<d/2$ gives $t\,\mu^2\le c\le C_c(1+z_+^2)$, the stated bound holds for all $t\ge 0$.
\end{proof}

\begin{proof}[\textbf{Proof of Lemma~\ref{lem:riskid}}]

Write $s:=s_{d,c}$ and $T:=T_d=\norm{X_d}^2$, and recall $\delta_{d,c}(x)=s(\norm{x}^2)\,x$. For a
vector field $F=(F_1,\dots,F_d):\mathbb{R}^d\to\mathbb{R}^d$, its \emph{divergence} is
$\mathrm{div}\,F:=\sum_{i=1}^d\partial F_i/\partial x_i$.

By Lemma~\ref{prop:beta} the posterior law of $Y$ given $T=t$ has density proportional to
$y^{c/2-1}(1-y)^{(d-c)/2-1}e^{ty/2}$, an exponential family with natural parameter $t/2$ and
sufficient statistic $y$. Differentiating $s(t)=\E[Y\mid T=t]$ under the integral sign  gives
\begin{equation}\label{eq:s't}
s'(t)=\tfrac12\,\mathrm{Var}(Y\mid T=t)=\tfrac12\bigl(q_d(t)-s(t)^2\bigr).    
\end{equation}
Since $Y\in(0,1)$, its conditional variance lies in $[0,1]$, so $|s'(t)|\le\tfrac12$ and
$s\in(0,1)$. Consequently $\delta_{d,c}$ is weakly differentiable with linear growth and
$\E\norm{\delta_{d,c}(X_d)-X_d}^2<\infty$.

For $X_d\sim N(\theta_d,I_d)$ and any weakly differentiable $g:\mathbb{R}^d\to\mathbb{R}^d$ with
$\E|\partial_i g_i(X_d)|<\infty$, integration by parts against the standard Gaussian density gives
Stein's identity \citep{Stein1981}
\[
\E\bigl[(X_{d,i}-\theta_{d,i})\,g_i(X_d)\bigr]=\E\bigl[\partial_i g_i(X_d)\bigr],
\quad i=1,\dots,d.
\]
Take $g(x):=\delta_{d,c}(x)-x$ and expand the risk by adding and subtracting $X_d$:
\[
\risk(\theta_d,\delta_{d,c})=\E\norm{\delta_{d,c}(X_d)-\theta_d}^2
=\E\norm{g(X_d)}^2
+2\,\E\bigl[(X_d-\theta_d)\cdot g(X_d)\bigr]
+\E\norm{X_d-\theta_d}^2.
\]
The last term equals $d$, and summing Stein's identity over $i$ turns the cross term into
$\E[\mathrm{div}\,g(X_d)]=\E[\mathrm{div}\,\delta_{d,c}(X_d)-d]$. Hence
\[
\risk(\theta_d,\delta_{d,c})
=\E\!\left[\norm{\delta_{d,c}(X_d)-X_d}^2 + 2\,\mathrm{div}\,\delta_{d,c}(X_d) - d\right].
\]

Now, with $t=\norm{x}^2$ and $\delta_{d,c}(x)-x=(s(t)-1)x$, we have
\[
\norm{\delta_{d,c}(x)-x}^2=(s(t)-1)^2\norm{x}^2=(s(t)-1)^2 t .
\]
Note that $\partial_i\bigl(s(\norm{x}^2)x_i\bigr)=s(t)+x_i\,s'(t)\,\partial_i\norm{x}^2
=s(t)+2x_i^2\,s'(t)$, so summing over $i$ and using $\sum_i x_i^2=t$, we obtain
\[
\mathrm{div}\,\delta_{d,c}(x)=d\,s(t)+2t\,s'(t).
\]

Substituting the above into the formula for $\risk(\theta_d,\delta_{d,c})$, we have 
\[
\risk(\theta_d,\delta_{d,c})
=
\E\!\left[(s(T)-1)^2T + 2d\,s(T) + 4T s'(T) - d\right].
\]
By \eqref{eq:s't}, $4T s'(T)=2T\bigl(q_d(T)-s(T)^2\bigr)$, then the moment identity
$T q_d(T)=c+(T-d)s(T)$ of Proposition~\ref{prop:moment}:
\begin{align*}
\risk(\theta_d,\delta_{d,c})
&=
\E\!\left[T-2T s(T)-T s(T)^2+2T q_d(T)+2d\,s(T)-d\right]\\
&=
\E\!\left[T-2T s(T)-T s(T)^2+2c+2(T-d)s(T)+2d\,s(T)-d\right]\\
&=
\E[T]-d + 2c - \E\!\left[T s(T)^2\right].
\end{align*}
Since $\E[T_d]=d+\norm{\theta_d}^2$, the claim follows.
\end{proof}

\subsection{Results in Sections~\ref{sec:weak}--\ref{sec:shellcritical}}
\label{app:weakproofs}

We first prove the more general Theorem \ref{thm:main} in Section \ref{sec:shellcritical}, and then apply it to derive Theorem \ref{thm:bounded} in Section \ref{sec:weak}. The proof of Theorem~\ref{thm:main} uses the exact finite-dimensional identity from Lemma~\ref{lem:riskid}. Thus it is enough to identify the limit of the single nonnegative term
\[
T_d s_{d,c}(T_d)^2
\]
and its expectation. For this purpose, we present three lemmas below, and then combine them to prove Theorem \ref{thm:main}.

Throughout this subsection, write
\[
\nu_d:=\norm{\theta_d}^2,
\quad
\beta_d:=\frac{\nu_d}{\sqrt d},
\quad
Z_d:=\frac{T_d-d}{\sqrt d}.
\]
We start by presenting the following lemma about the limit behavior of $Z_d$.

\begin{lemma}[Limit Distribution of $Z_d$]
\label{lem:Zd}
If $\beta_d\to\beta\in[0,\infty)$, then
\[
Z_d\Rightarrow Z_\beta\sim N(\beta,2),
\quad
\sup_d\E|Z_d|^4<\infty.
\]
\end{lemma}

\begin{proof}[\textbf{Proof of Lemma~\ref{lem:Zd}}]
Write $X_d=\theta_d+\varepsilon_d$, where $\varepsilon_d\sim N_d(0,I_d)$. Then
\[
Z_d=A_d+B_d+\beta_d,
\quad
A_d:=\frac{\norm{\varepsilon_d}^2-d}{\sqrt d},
\quad
B_d:=\frac{2\theta_d^\top\varepsilon_d}{\sqrt d}.
\]
Since $\norm{\varepsilon_d}^2 \sim \chi^2_d$, by the central limit theorem we have $A_d\Rightarrow N(0,2)$. Moreover, $B_d$ is Gaussian with mean $0$ and
\[
\Var(B_d)=\frac{4\nu_d}{d}=\frac{4\beta_d}{\sqrt d}\to0,
\]
so $B_d\to0$ in probability. Then, since $\beta_d\to\beta$ too, we have by the Slutsky's theorem 
$$Z_d\Rightarrow Z_\beta\sim N(\beta,2).$$

For the fourth moment, the centered fourth moment of a $\chi_d^2$ variable is $12d^2+48d$, and therefore
\[
\E A_d^4=12+\frac{48}{d}.
\]
Also,
\[
\E B_d^4=3\Var(B_d)^2=\frac{48\beta_d^2}{d}.
\]
The sequence $\{\beta_d\}$ is bounded. Hence, using
$(x+y+z)^4\le27(x^4+y^4+z^4)$,
\[
\sup_d\E|Z_d|^4
\le
27\left(
\sup_d\E A_d^4+
\sup_d\E B_d^4+
\sup_d|\beta_d|^4
\right)<\infty.
\]
\end{proof}

The next lemma establishes that the conditional expectation function $h_c(z)$ defined in Theorem \ref{thm:main} emerges as the limit of $\sqrt d\,s_{d,c}(d+\sqrt d\,z)$, which governs the limit behavior of $\sqrt{T_d} s_{d,c}(T_d)$.

For $b>0$, define
\[
I_b(z):=\int_0^\infty
u^{b-1}\exp\!\left(\frac{zu}{2}-\frac{u^2}{4}\right)du,
\quad
h_c(z):=\frac{I_{c/2+1}(z)}{I_{c/2}(z)}.
\]
This agrees with the definition in Theorem~\ref{thm:main}, since $h_c(z)$ is the mean of the probability density proportional to
$u^{c/2-1}\exp(zu/2-u^2/4)$ on $(0,\infty)$.

\begin{lemma}[Emergence of $h_c(z)$ in the limit]
\label{lem:boundary}
Fix $c>0$. For every $M<\infty$,
\[
\sup_{|z|\le M}
\left|
\sqrt d\,s_{d,c}(d+\sqrt d\,z)-h_c(z)
\right|\to0.
\]
\end{lemma}

\begin{proof}[\textbf{Proof of Lemma~\ref{lem:boundary}}]
Set $b=c/2$ and $t=d+\sqrt d\,z$. By Lemma~\ref{prop:beta}, the posterior mean
$s_{d,c}(t)=\E[Y\mid T_d=t]$ is the ratio of the first posterior moment to the normalizing
constant:
\[
s_{d,c}(t)
=
\frac{A_{d,1}(t)}{A_{d,0}(t)},
\]
where, for $m=0,1$,
\[
A_{d,m}(t)
:=
\int_0^1
y^{b-1+m}(1-y)^{(d-c)/2-1}
\exp\!\left(\frac{ty}{2}\right)dy.
\]
We now evaluate this ratio on the local scale $t=d+\sqrt d\,z$. Put
$y=u/\sqrt d$, so that $dy=du/\sqrt d$ and $u\in(0,\sqrt d)$. Then
\[
A_{d,m}(d+\sqrt d\,z)
=d^{-(b+m)/2}J_{d,m}(z),
\]
with
\[
J_{d,m}(z)
:=
\int_0^{\sqrt d}
 u^{b-1+m}
\left(1-\frac{u}{\sqrt d}\right)^{(d-c)/2-1}
\exp\!\left(\frac{(d+\sqrt d\,z)u}{2\sqrt d}\right)du.
\]
Therefore
\[
\sqrt d\,s_{d,c}(d+\sqrt d\,z)
=
\frac{J_{d,1}(z)}{J_{d,0}(z)}.
\]

Define, for $0<u<\sqrt d$,
\[
\phi_d(u,z)
:=
\left(\frac{d-c}{2}-1\right)
\log\!\left(1-\frac{u}{\sqrt d}\right)
+
\frac{(d+\sqrt d\,z)u}{2\sqrt d},
\]
which  is the exponent of the only part of the integrand depending on $d$ in a non-polynomial way. For
fixed $u>0$, we use the expansion $\log(1-r)=-r-r^2/2+O(r^3)$ with $r=u/\sqrt d$ to derive
\[
\begin{aligned}
\phi_d(u,z)
&=\left(\frac{d-c}{2}-1\right)
\left(-\frac{u}{\sqrt d}-\frac{u^2}{2d}+O(d^{-3/2})\right)
+\frac{\sqrt d\,u}{2}+\frac{zu}{2} \\
&=\frac{zu}{2}-\frac{u^2}{4}+O(d^{-1/2}),
\end{aligned}
\]
where the linear terms $-\sqrt d\,u/2$ and $\sqrt d\,u/2$ cancel. Hence
\[
\sup_{|z|\le M}
\left|
\phi_d(u,z)-\left(\frac{zu}{2}-\frac{u^2}{4}\right)
\right|\to0.
\]
In fact, the remainder is independent of $z$.

It remains to justify integration uniformly over $|z|\le M$. For sufficiently large $d$, the inequality
$\log(1-r)\le-r-r^2/2$, $0<r<1$, gives
\begin{align*}
\phi_d(u,z)
&\le
\frac{Mu}{2}
+
\frac{c+2}{2\sqrt d}u
-
\frac{d-c-2}{4d}u^2 \\
&\le
\frac{Mu}{2}+u-\frac{u^2}{8},
\quad 0<u<\sqrt d,
\end{align*}
where the last inequality holds once $d$ is large enough in terms of $c$. Extend the integrand defining $J_{d,m}$ to $(0,\infty)$ by setting it equal to zero for $u\ge\sqrt d$. For each fixed $u>0$, its convergence to
\[
u^{b-1+m}\exp\!\left(\frac{zu}{2}-\frac{u^2}{4}\right)
\]
is uniform over $|z|\le M$, and the supremum of the absolute difference is bounded by a constant multiple of
\[
u^{b-1+m}\exp\!\left(\frac{Mu}{2}+u-\frac{u^2}{8}\right),
\]
which is integrable on $(0,\infty)$. By dominated convergence, we have
\[
\sup_{|z|\le M}
\left|J_{d,m}(z)-I_{b+m}(z)\right|\to0,
\quad m=0,1.
\]
The function $I_b$ is continuous and strictly positive, and hence
$\inf_{|z|\le M}I_b(z)>0$. Taking ratios proves the proposition.
\end{proof}

Finally, we present a lemma that summarizes some useful properties of the limit density $h_c$ that will be useful in the proof of Theorem \ref{thm:main}.

\begin{lemma}[Some Properties of $h_c$]
\label{lem:profile}
Let $c>0$.
\begin{enumerate}
\item[(i)] The function $h_c$ is continuously differentiable and satisfies
\begin{equation}\label{eq:ricatti_id}
h_c'(z)
=
\frac c2+\frac z2h_c(z)-\frac12h_c(z)^2.
\end{equation}
\item[(ii)] For every $z\in\mathbb R$,
\[
h_c(z)^2\le c+zh_c(z).
\]
Consequently,
\[
0<h_c(z)\le\sqrt c\quad \text{for } z\le0,
\quad
0<h_c(z)\le z+\sqrt c\quad \text{for }z\ge0,
\]
and therefore $h_c(z)=O(1+z_+)$ and $h_c'(z)=O(1+z^2)$ as $|z|\to\infty$.
\item[(iii)] If $Z_\beta\sim N(\beta,2)$ and
\[
m_c(\beta):=\E[h_c(Z_\beta)],
\quad
q_c(\beta):=\E[h_c(Z_\beta)^2],
\]
then these expectations are finite and
\[
q_c(\beta)=c+\beta m_c(\beta).
\]
In particular,
\[
L_c(\beta)
=q_c(\beta)-2\beta m_c(\beta)
=2c-q_c(\beta)
=c-\beta m_c(\beta),
\quad
L_c(0)=c.
\]
\end{enumerate}
\end{lemma}

\begin{proof}[\textbf{Proof of Lemma~\ref{lem:profile}}]
Write $b=c/2$. Differentiating under the integral sign, which is justified locally uniformly in $z$ by Gaussian domination, we obtain
\[
I_b'(z)=\frac12 I_{b+1}(z).
\]
We next record the integration-by-parts identity used below. Since $b>0$, the function
\[
u^b\exp\!\left(\frac{zu}{2}-\frac{u^2}{4}\right)
\]
vanishes as $u\downarrow0$; the Gaussian factor also makes it vanish as $u\to\infty$. Hence the integral of its derivative over $(0,\infty)$ is zero. Computing that derivative gives
\[
0
=\int_0^\infty
\left(bu^{b-1}+\frac z2u^b-\frac12u^{b+1}\right)
\exp\!\left(\frac{zu}{2}-\frac{u^2}{4}\right)\,du.
\]
In terms of the integrals $I_a(z)$, this is
\[
0=bI_b(z)+\frac z2 I_{b+1}(z)-\frac12 I_{b+2}(z),
\]
or equivalently
\[
I_{b+2}(z)=zI_{b+1}(z)+2bI_b(z).
\]
Finally, differentiating $h_c=I_{b+1}/I_b$ by the quotient rule and using $I_a'(z)=I_{a+1}(z)/2$ gives
\[
h_c'(z)
=
\frac12\frac{I_{b+2}(z)}{I_b(z)}
-
\frac12\left(\frac{I_{b+1}(z)}{I_b(z)}\right)^2
=
\frac12\left(z\frac{I_{b+1}(z)}{I_b(z)}+2b\right)
-\frac12h_c(z)^2
=
\frac c2+\frac z2h_c(z)-\frac12h_c(z)^2,
\]
proving part (i).

Let $U_z$ have density proportional to
$u^{b-1}\exp(zu/2-u^2/4)$ on $(0,\infty)$. Then
\[
\E U_z=h_c(z),
\quad
\E U_z^2
=
\frac{I_{b+2}(z)}{I_b(z)}
=c+zh_c(z).
\]
By Jensen's inequality  $h_c(z)^2\le c+zh_c(z)$. If $z\le0$, this implies $h_c(z)^2\le c$. If $z\ge0$, solving the quadratic inequality gives
\[
h_c(z)
\le
\frac{z+\sqrt{z^2+4c}}{2}
\le z+\sqrt c.
\]
Hence, $h_c(z)=O(1+z_+)$, and \eqref{eq:ricatti_id} then yields $h_c'(z)=O(1+z^2)$. This proves part (ii).

The growth bounds above imply that
$\E|Z_\beta h_c(Z_\beta)|$, $\E|h_c'(Z_\beta)|$, and $\E[h_c(Z_\beta)^2]$ are finite. Let
\[
\varphi_\beta(z):=(4\pi)^{-1/2}\exp\!\left\{-\frac{(z-\beta)^2}{4}\right\}
\]
denote the density of $Z_\beta\sim N(\beta,2)$. Then
\[
\varphi_\beta'(z)=-\frac{z-\beta}{2}\varphi_\beta(z),
\quad\text{or equivalently}\quad
(z-\beta)\varphi_\beta(z)=-2\varphi_\beta'(z).
\]
Therefore, for $A>0$,
\[
\int_{-A}^A (z-\beta)h_c(z)\varphi_\beta(z)\,dz
=
-2\int_{-A}^A h_c(z)\varphi_\beta'(z)\,dz
=
-2\big[h_c(z)\varphi_\beta(z)\big]_{-A}^A
+2\int_{-A}^A h_c'(z)\varphi_\beta(z)\,dz.
\]
By part (ii), $h_c$ has at most linear growth and $h_c'$ has at most quadratic growth, while
$\varphi_\beta$ has Gaussian tails. Hence
\[
\big[h_c(z)\varphi_\beta(z)\big]_{-A}^A\to0
\quad\text{as }A\to\infty,
\]
and dominated convergence gives
\[
\E[(Z_\beta-\beta)h_c(Z_\beta)]
=2\E[h_c'(Z_\beta)].
\]
Using part (i) on the right-hand side gives
\[
\E[Z_\beta h_c(Z_\beta)]-\beta m_c(\beta)
=
c+\E[Z_\beta h_c(Z_\beta)]-q_c(\beta)
\]
which implies $q_c(\beta)=c+\beta m_c(\beta)$. The remaining identities in part (iii) follow algebraically.
\end{proof}


We now apply Lemmas \ref{lem:Zd}, \ref{lem:boundary}  and \ref{lem:profile} together to prove Theorem \ref{thm:main}.

\begin{proof}[\textbf{Proof of Theorem~\ref{thm:main}}]
For $z\ge-\sqrt d$, define
\[
v_d(z):=\sqrt d\,s_{d,c}(d+\sqrt d\,z),
\quad
G_d(z):=\left(1+\frac{z}{\sqrt d}\right)v_d(z)^2,
\]
and extend $G_d$ arbitrarily to $z<-\sqrt d$. Since $T_d=d+\sqrt d\,Z_d$ and $Z_d\ge-\sqrt d$ almost surely,
\[
G_d(Z_d)=T_d s_{d,c}(T_d)^2.
\]
Lemma~\ref{lem:boundary} implies that, on every compact interval,
$v_d\to h_c$ uniformly. Hence $G_d\to h_c^2$ uniformly on every compact interval as well.

By Lemma~\ref{lem:Zd}, $Z_d\Rightarrow Z_\beta$, and in particular $\{Z_d\}$ is tight. Fix $\varepsilon>0$ and choose $M$ so that
$\sup_d\Pbb(|Z_d|>M)<\varepsilon$. Then
\[
|G_d(Z_d)-h_c(Z_d)^2|\mathbf 1_{\{|Z_d|\le M\}}
\le
\sup_{|z|\le M}|G_d(z)-h_c(z)^2|
\to0.
\]
Thus $G_d(Z_d)-h_c(Z_d)^2\to0$ in probability. Since $h_c^2$ is continuous and $Z_d\Rightarrow Z_\beta$, the continuous mapping theorem and Slutsky's theorem give
\[
T_d s_{d,c}(T_d)^2
=G_d(Z_d)
\Rightarrow h_c(Z_\beta)^2.
\]

Then, by Proposition~\ref{prop:moment}, we have
\[
0\le T_d s_{d,c}(T_d)^2
\le C_c\{1+(Z_d^+)^2\}.
\]
and, by Lemma~\ref{lem:Zd},
\[
\sup_d\E\left[\left(1+(Z_d^+)^2\right)^2\right]<\infty.
\]
Hence the dominating family $1+(Z_d^+)^2$ is bounded in $L^2$ and is therefore uniformly integrable. Similarly, 
$\{T_d s_{d,c}(T_d)^2\}$ is also uniformly integrable. Convergence in distribution together with uniform integrability yields
\[
\E[T_d s_{d,c}(T_d)^2]
\to
q_c(\beta):=\E[h_c(Z_\beta)^2].
\]

Finally, Lemma~\ref{lem:riskid} gives the exact identity
\[
\risk(\theta_d,\delta_{d,c})-\norm{\theta_d}^2
=
2c-\E[T_d s_{d,c}(T_d)^2].
\]
Therefore
\[
\risk(\theta_d,\delta_{d,c})-\norm{\theta_d}^2
\to
2c-q_c(\beta)
=L_c(\beta),
\]
where the last equality is Lemma~\ref{lem:profile}(iii).
\end{proof}

\begin{proof}[\textbf{Proof of Theorem~\ref{thm:bounded}}]
If $\norm{\theta_d}^2\to\nu<\infty$, then
$\norm{\theta_d}^2/\sqrt d\to0$. Theorem~\ref{thm:main} and Lemma~\ref{lem:profile}(iii) give
\[
\risk(\theta_d,\delta_{d,c})-\norm{\theta_d}^2
\to
L_c(0)=c.
\]
Combining this with $\norm{\theta_d}^2\to\nu$ proves
$\risk(\theta_d,\delta_{d,c})\to\nu+c$.
\end{proof}

\subsection{Results in Section~\ref{sec:dense}}
\label{app:denseproofs}

To prove Theorem~\ref{thm:strongfirst} and Theorem~\ref{thm:dense2}, we first introduce the following notation and lemmas.

Let $\tau$ denote the normalized squared norm $T_d/d$, not the prior SD scale.  Fix a $\tau>1$ and condition on the  event $\{T_d=d\tau\}$. Write $\Pi_{d,c}(\cdot\mid T_d=t)$ for the posterior law of the latent variance $g$ given $T_d=t$ under the prior $\pi_{d,c}$. For $\tau>1$, define
\[
\Phi_\tau(g):=\log(1+g)+\frac{\tau}{1+g},
\quad
g_\tau:=\tau-1,
\quad
a(\tau):=\frac{g_\tau}{1+g_\tau}=1-\frac1\tau.
\]
Since $\Phi_\tau'(g)=(g+1-\tau)/(1+g)^2$, the point $g_\tau$ is the unique minimizer of $\Phi_\tau$ on $(0,\infty)$. 

The next lemma shows that the posterior for $g$ concentrates near this minimizer uniformly over compact $\tau$-sets. Its role is to convert the \emph{random} Bayes shrinkage factor $s_{d,c}(T_d)$ into the \emph{deterministic} function $a(\tau_d)$ of the normalized value $\tau_d:=T_d/d$.
\begin{lemma}[Uniform posterior concentration in the strong regime]
\label{lem:densepost}
Let $K\subset (1,\infty)$ be compact. For every $\varepsilon>0$,
\[
\sup_{\tau\in K}
\Pi_{d,c}\!\left(
\left.
\left|g-g_\tau\right|>\varepsilon
\;\right|\;
T_d=d\tau
\right)
\to 0.
\]
Consequently,
\[
\sup_{\tau\in K}\left|s_{d,c}(d\tau)-a(\tau)\right|\to 0.
\]
\end{lemma}

\begin{proof}[\textbf{Proof of Lemma~\ref{lem:densepost}}]
Fix $\varepsilon>0$. Since $K$ is compact and $\tau\mapsto g_\tau=\tau-1$ is continuous,
$g_\tau$ ranges over a compact interval $[m_K,M_K]\subset(0,\infty)$. We can choose a small enough $\delta>0$ so
that 
$I_\tau:=[g_\tau-\delta,g_\tau+\delta]\subset(0,\infty)$ 
for every $\tau\in K$, and $\delta<\varepsilon/2$.

Conditional on $T_d=d\tau$, the posterior density of $g$ is proportional to
\[
g^{c/2-1}\exp\!\left(-\frac d2 \Phi_\tau(g)\right), \quad g>0.
\]
Since $\Phi_\tau$ is continuous and has the unique minimizer $g_\tau$, there exists $\eta=\eta(K,\varepsilon)>0$ s.t. 
\[
\Phi_\tau(g)\ge \Phi_\tau(g_\tau)+3\eta
\]
for all $\tau\in K$ and all $g\in[0,M]\setminus (g_\tau-\varepsilon,g_\tau+\varepsilon)$, where
$M> M_K+\varepsilon$ is chosen below. Also, since $\Phi_\tau(g)\to\infty$ as $g\to\infty$
uniformly over $\tau\in K$, we may choose $M$ large enough that
\begin{equation}\label{eq:Phi_gap}
 \inf_{\tau\in K}\inf_{g\ge M}\Phi_\tau(g)\ge \sup_{\tau\in K}\Phi_\tau(g_\tau)+3\eta.   
\end{equation}
If necessary, we may set $M$ to be even larger  so that
\[
\log M \ge \sup_{\tau\in K}\Phi_\tau(g_\tau)+2\eta.
\]

Let
\[
N_{d,\tau}(\varepsilon)
:=
\int_{|g-g_\tau|>\varepsilon}
g^{c/2-1}\exp\!\left(-\frac d2 \Phi_\tau(g)\right)\,dg
\]
and
\[
D_{d,\tau}
:=
\int_0^\infty
g^{c/2-1}\exp\!\left(-\frac d2 \Phi_\tau(g)\right)\,dg.
\]

For $D_{d,\tau}$, continuity of $g^{c/2-1}$ on the compact union of the intervals $I_\tau$
implies a uniform lower bound $g^{c/2-1}\ge b_K>0$ there. Also, by continuity of $\Phi_\tau$,
\[
\Phi_\tau(g)\le \Phi_\tau(g_\tau)+\eta
\quad
\text{for } g\in I_\tau,\ \tau\in K
\]
after shrinking $\delta$ if necessary. Hence
\[
D_{d,\tau}
\ge
\int_{I_\tau}
g^{c/2-1}\exp\!\left(-\frac d2 \Phi_\tau(g)\right)\,dg
\ge
2\delta\, b_K \exp\!\left(-\frac d2(\Phi_\tau(g_\tau)+\eta)\right).
\]

For $N_{d,\tau}$, note that
\[
N_{d,\tau}(\varepsilon)\le N_{d,\tau}^{(1)}(\varepsilon)+N_{d,\tau}^{(2)}(\varepsilon),
\]
where the first term integrates over $[0,M]$ and the second over $[M,\infty)$. On $[0,M]$,
$g^{c/2-1}$ is integrable and the uniform gap \eqref{eq:Phi_gap} gives
\[
N_{d,\tau}^{(1)}(\varepsilon)
\le
\left(\int_0^M g^{c/2-1}\,dg\right)
\exp\!\left(-\frac d2(\Phi_\tau(g_\tau)+3\eta)\right).
\]
On $[M,\infty)$, use $\Phi_\tau(g)\ge \log(1+g)$ to get
\[
N_{d,\tau}^{(2)}(\varepsilon)
\le
\int_M^\infty g^{c/2-1}(1+g)^{-d/2}\,dg
\le
\int_M^\infty g^{c/2-1-d/2}\,dg
=
\frac{M^{c/2-d/2}}{d/2-c/2},
\]
which is valid for $d>c$.

We now divide $N_{d,\tau}^{(1)}$ and $N_{d,\tau}^{(2)}$  by the $D_{d,r}$ lower bound.  
For the first piece, write $C_K:=\int_0^M g^{c/2-1}\,dg<\infty$, and obtain
\[
\frac{N_{d,\tau}^{(1)}(\varepsilon)}{D_{d,\tau}}
\le
\frac{C_K}{2\delta\,b_K}\,
\exp\!\left(-\frac d2\bigl[(\Phi_\tau(g_\tau)+3\eta)-(\Phi_\tau(g_\tau)+\eta)\bigr]\right)
=
\frac{C_K}{2\delta\,b_K}\,e^{-d\eta}.
\]
For the second piece, using $M^{(c-d)/2}=M^{c/2}\exp\!\left(-\tfrac d2\log M\right)$,
\[
\frac{N_{d,\tau}^{(2)}(\varepsilon)}{D_{d,\tau}}
\le
\frac{M^{c/2}}{(d-c)\,\delta\,b_K}\,
\exp\!\left(-\frac d2\bigl[\log M-\Phi_\tau(g_\tau)-\eta\bigr]\right).
\]
By the choice $\log M\ge\sup_{\tau\in K}\Phi_\tau(g_\tau)+2\eta$, the bracket satisfies
$\log M-\Phi_\tau(g_\tau)-\eta\ge\eta$ for every $\tau\in K$, so this ratio is at most
$\dfrac{M^{c/2}}{(d-c)\,\delta\,b_K}\,e^{-d\eta/2}$.

In both displayed inequalities above the constants $C_K,\delta,b_K,\eta,M$ depend only on $K$ and $\varepsilon$, not on
$\tau$. Adding the two bounds therefore gives
\[
\sup_{\tau\in K}
\Pi_{d,c}\!\left(
\left.
\left|g-g_\tau\right|>\varepsilon
\;\right|\;
T_d=d\tau
\right)
\le
\frac{C_K}{2\delta\,b_K}\,e^{-d\eta}
+
\frac{M^{c/2}}{(d-c)\,\delta\,b_K}\,e^{-d\eta/2}
\;\to\;0
\]
as $d\to\infty$. 

For the posterior mean, note that $\psi(g):=g/(1+g)$ is bounded by $1$ and Lipschitz with constant
at most $1$. Therefore
\[
\left|s_{d,c}(d\tau)-a(\tau)\right|
=
\left|
\E\!\left[\left.\psi(g)-\psi(g_\tau)\right|T_d=d\tau\right]
\right|
\le
\varepsilon
+
\Pi_{d,c}\!\left(
\left.
\left|g-g_\tau\right|>\varepsilon
\;\right|\;
T_d=d\tau
\right).
\]
Taking the supremum over $\tau\in K$ and then letting $d\to\infty$ proves the second claim.
\end{proof}

Lemma~\ref{lem:densepost} is the strong-signal analogue of the local near-zero approximation in the critical regime: in the strong regime, $g$ is no longer pinned near $0$, and the posterior mean $\E[Y\mid T_d=d\tau]$ behaves like the ridge shrinkage coefficient $a(\tau)=1-1/\tau$.

To turn this conditional approximation into an unconditional risk statement, we also need to know where $\tau_d=T_d/d$ typically lies under the model $X_d\sim N_d(\theta_d,I_d)$. The next lemma is a law of large numbers for the macroscopic observables: it identifies the deterministic normalized value $\tau_0=1+\rho$ around which the Laplace point $g_{\tau_d}$ fluctuates.
\begin{lemma}[Macroscopic observation limits]
\label{lem:densestat}
If $\norm{\theta_d}^2/d\to \rho>0$, then
\[
\frac{T_d}{d}\to 1+\rho
\quad\text{and}\quad 
\frac{\theta_d^\top X_d}{d}\to \rho \text{ in } L^1.
\]

\end{lemma}

\begin{proof}[\textbf{Proof of Lemma~\ref{lem:densestat}}]
Write $\nu_d=\norm{\theta_d}^2$. Then
\[
\E\!\left[\frac{T_d}{d}\right]=1+\frac{\nu_d}{d}\to 1+\rho,
\quad
\Var\!\left(\frac{T_d}{d}\right)=\frac{2d+4\nu_d}{d^2}\to 0.
\]
Hence $T_d/d\to 1+\rho$ in $L^2$, and therefore in $L^1$.

Also
\[
\theta_d^\top X_d = \nu_d + \theta_d^\top \varepsilon_d,
\quad
\varepsilon_d\sim N_d(0,I_d),
\]
so
\[
\E\!\left[\frac{\theta_d^\top X_d}{d}\right]=\frac{\nu_d}{d}\to \rho,
\quad
\Var\!\left(\frac{\theta_d^\top X_d}{d}\right)=\frac{\nu_d}{d^2}\to 0.
\]
Again this yields $L^2$, hence $L^1$, convergence.
\end{proof}

Theorem~\ref{thm:strongfirst} follows by combining Lemma~\ref{lem:riskid} with Lemma~\ref{lem:densepost}--\ref{lem:densestat}: on the event $\{T_d/d\in K\}$, the shrinkage factor $s_{d,c}(T_d)$ is close to $a(T_d/d)$, and $T_d/d$ concentrates near $\tau_0=1+\rho$, so the Bayes rule behaves like constant ridge shrinkage with coefficient $a(\tau_0)=\rho/(1+\rho)$.

For the refined second-order expansion, concentration is not sufficient: we need a \emph{one-step} Laplace expansion of $s_{d,c}(d\tau)$ that controls the $1/d$ correction uniformly over compact $\tau$-sets. This expansion is the strong-signal counterpart of the local near-zero limit (Lemma~\ref{lem:boundary}) and is the main input for the $O(1)$ term in Theorem~\ref{thm:dense2}.
\begin{lemma}[Uniform one-step Laplace expansion]
\label{lem:denseLaplace}
Let $K\subset(1,\infty)$ be compact. Then
\[
\sup_{\tau\in K}
d\left|
s_{d,c}(d\tau)-a(\tau)-\frac{b_c(\tau)}{d}
\right|
\to 0,
\quad
b_c(\tau):=\frac{2}{\tau}+\frac{c-2}{\tau-1}.
\]
\end{lemma}

\begin{proof}[\textbf{Proof of Lemma~\ref{lem:denseLaplace}}]
Fix compact $K\subset(1,\infty)$. Write
\[
\psi(g):=\frac{g}{1+g},
\quad
A(g):=g^{c/2-1},
\quad
\phi_\tau(g):=\frac12\Phi_\tau(g)=\frac12\log(1+g)+\frac{\tau}{2(1+g)}.
\]
Then
\[
s_{d,c}(d\tau)
=
\frac{N_{d,\tau}}{D_{d,\tau}},
\quad
N_{d,\tau}:=\int_0^\infty \psi(g)A(g)e^{-d\phi_\tau(g)}\,dg,
\quad
D_{d,\tau}:=\int_0^\infty A(g)e^{-d\phi_\tau(g)}\,dg.
\]
As in the proof of Lemma~\ref{lem:densepost}, $\phi_\tau$ has the unique minimizer
$g_\tau=\tau-1$, and there exist $\eta,\omega_0>0$ such that, uniformly over $\tau\in K$, the
contributions to $N_{d,\tau}$ and $D_{d,\tau}$ from $|g-g_\tau|>\eta$ are
$O(e^{-d\omega_0})$ relative to the main Laplace factor.

Hence it suffices to integrate over $|g-g_\tau|\le \eta$ and set
\[
g=g_\tau+\frac{u}{\sqrt d}.
\]
Write $\phi_{\tau,j}:=\phi_\tau^{(j)}(g_\tau)$ and similarly let $A_\tau$, $A_\tau'$, $\psi_\tau$,
$\psi_\tau'$, and $\psi_\tau''$ denote derivatives evaluated at $g_\tau$.

Let $M_d=L\sqrt{\log d}$, with $L>0$ chosen below. After decreasing $\eta$ if necessary,
compactness of $K$ gives constants $a,C>0$ such that, uniformly in $\tau\in K$,
\[
\phi_\tau(g)-\phi_\tau(g_\tau)\ge a(g-g_\tau)^2,
\quad |g-g_\tau|\le\eta,
\]
and the derivatives appearing below are uniformly bounded on this neighborhood. Hence, for every
fixed $m$,
\[
\int_{M_d<|u|\le\eta\sqrt d}(1+|u|^m)
\exp\!\left[-d\left\{\phi_\tau\!\left(g_\tau+\frac{u}{\sqrt d}\right)
-\phi_\tau(g_\tau)\right\}\right]du
\le C\int_{|u|>M_d}(1+|u|^m)e^{-a u^2}\,du
=o(d^{-1})
\]
uniformly in $\tau\in K$ once $L$ is sufficiently large. On $|u|\le M_d$, we have,
uniformly in $\tau\in K$,
\begin{align*}
d\left\{\phi_\tau\!\left(g_\tau+\frac{u}{\sqrt d}\right)-\phi_\tau(g_\tau)\right\}
&=
\frac{\phi_{\tau,2}u^2}{2}
+\frac{\phi_{\tau,3}u^3}{6\sqrt d}
+\frac{\phi_{\tau,4}u^4}{24d}
+O_K\!\left(\frac{|u|^5}{d^{3/2}}\right),\\
A\!\left(g_\tau+\frac{u}{\sqrt d}\right)
&=
A_\tau+\frac{A_\tau' u}{\sqrt d}
+\frac{A_\tau''u^2}{2d}
+O_K\!\left(\frac{|u|^3}{d^{3/2}}\right),\\
\psi\!\left(g_\tau+\frac{u}{\sqrt d}\right)
&=
\psi_\tau+\frac{\psi_\tau' u}{\sqrt d}
+\frac{\psi_\tau''u^2}{2d}
+O_K\!\left(\frac{|u|^3}{d^{3/2}}\right).
\end{align*}
Since $M_d=O(\sqrt{\log d})$, the first remainder is
$O_K((\log d)^{5/2}d^{-3/2})=o_K(d^{-1})$ and the remaining two remainders are
$O_K((\log d)^{3/2}d^{-3/2})=o_K(d^{-1})$. The quadratic lower bound gives a uniform Gaussian
envelope, so these expansions may be integrated term by term.

Multiplying the expansions and integrating wrt $e^{-\phi_{\tau,2}u^2/2}$ yields
\[
\begin{aligned}
N_{d,\tau}
&=
C_{d,\tau}
\left[
A_\tau\psi_\tau M_{0,\tau}
+
\frac{\Gamma_{N}(\tau)}{d}
+
o_K(d^{-1})
\right],\\
D_{d,\tau}
&=
C_{d,\tau}
\left[
A_\tau M_{0,\tau}
+
\frac{\Gamma_{D}(\tau)}{d}
+
o_K(d^{-1})
\right],
\end{aligned}
\]
uniformly over $\tau\in K$, where
\[
C_{d,\tau}:=d^{-1/2}e^{-d\phi_\tau(g_\tau)}
\]
and
\[
M_{k,\tau}:=\int_{\mathbb R}u^k e^{-\phi_{\tau,2}u^2/2}\,du.
\]
Expanding the quotient first gives
\[
\frac{N_{d,\tau}}{D_{d,\tau}}
=
\psi_\tau
+
\frac{1}{d}
\frac{\Gamma_N(\tau)-\psi_\tau\Gamma_D(\tau)}{A_\tau M_{0,\tau}}
+o_K(d^{-1}).
\]
In the difference $\Gamma_N-\psi_\tau\Gamma_D$, the only  surviving terms are
\[
\Gamma_N(\tau)-\psi_\tau\Gamma_D(\tau)
=
\left(A_\tau'\psi_\tau'+\frac{A_\tau\psi_\tau''}{2}\right)M_{2,\tau}
-
\frac{A_\tau\psi_\tau'\phi_{\tau,3}}{6}M_{4,\tau}.
\]
Using the Gaussian moments $M_{2,\tau}/M_{0,\tau}=1/\phi_{\tau,2}$ and
$M_{4,\tau}/M_{0,\tau}=3/\phi_{\tau,2}^2$, we deduce  the one-step Laplace ratio formula
\[
s_{d,c}(d\tau)
=
\psi_\tau
+
\frac{1}{d}
\left\{
\frac{\psi_\tau''}{2\phi_{\tau,2}}
+
\frac{A_\tau'}{A_\tau}\frac{\psi_\tau'}{\phi_{\tau,2}}
-
\frac{\phi_{\tau,3}\psi_\tau'}{2\phi_{\tau,2}^2}
\right\}
+
o_K(d^{-1}).
\]
Now
\[
\psi_\tau=a(\tau)=1-\frac1\tau,
\quad
\psi_\tau'=\frac1{\tau^2},
\quad
\psi_\tau''=-\frac{2}{\tau^3},
\]
while
\[
\phi_{\tau,2}=\frac{1}{2\tau^2},
\quad
\phi_{\tau,3}=-\frac{2}{\tau^3},
\quad
\frac{A_\tau'}{A_\tau}=\frac{c/2-1}{\tau-1}.
\]
Substituting these derivatives gives
\[
\frac{\psi_\tau''}{2\phi_{\tau,2}}
+
\frac{A_\tau'}{A_\tau}\frac{\psi_\tau'}{\phi_{\tau,2}}
-
\frac{\phi_{\tau,3}\psi_\tau'}{2\phi_{\tau,2}^2}
=
-\frac{2}{\tau}
+
\frac{c-2}{\tau-1}
+
\frac{4}{\tau}
=
\frac{2}{\tau}+\frac{c-2}{\tau-1},
\]
which is exactly $b_c(\tau)$.
\end{proof}

Lemma~\ref{lem:denseLaplace} isolates the only place where the prior exponent $c$ enters at order $1/d$, via the explicit term $(c-2)/(\tau-1)$. When combined with the $+2c$ term in Lemma~\ref{lem:riskid}, this structure makes the strong-signal $O(1)$ risk correction tractable.

We now present the main proofs of Theorems~\ref{thm:strongfirst} and \ref{thm:dense2}. 

\begin{proof}[\textbf{Proof of Theorem~\ref{thm:strongfirst}}]
Set $\tau:=1+\rho$ and $a_\rho:=\rho/(1+\rho)=a(\tau)$. Let $K\subset (1,\infty)$ be a compact interval
 containing $\tau$. By Lemma~\ref{lem:densestat}, 
$T_d/d\to \tau$ 
in probability. Therefore
\[
\Pbb\!\left(\frac{T_d}{d}\notin K\right)\to 0.
\]
On the event $\{T_d/d\in K\}$,
\[
\left|s_{d,c}(T_d)-a\!\left(\frac{T_d}{d}\right)\right|
\le
\sup_{u\in K}\left|s_{d,c}(du)-a(u)\right|,
\]
which converges to $0$ by Lemma~\ref{lem:densepost}. Since $a(u)=1-1/u$ is continuous,
$a(T_d/d)\to a_\rho$ in probability, so 
$$s_{d,c}(T_d)\to a_\rho$$ 
in probability. Since $0\le s_{d,c}(T_d)\le 1$, this convergence also holds in $L^1$ for both
$s_{d,c}(T_d)$ and $s_{d,c}(T_d)^2$.

Now divide the exact risk decomposition in Lemma \ref{lem:riskid} by $d$:
\[
\frac{\risk(\theta_d,\delta_{d,c})}{d}
=
\E\!\left[s_{d,c}(T_d)^2\frac{T_d}{d}\right]
-2\E\!\left[s_{d,c}(T_d)\frac{\theta_d^\top X_d}{d}\right]
+
\frac{\norm{\theta_d}^2}{d}.
\]
For the first term,
\[
\E\!\left|s_{d,c}(T_d)^2\frac{T_d}{d}-a_\rho^2\tau\right|
\le
\E\!\left|\frac{T_d}{d}-\tau\right|
+
\tau\,\E\!\left|s_{d,c}(T_d)^2-a_\rho^2\right|
\to 0,
\]
using $0\le s_{d,c}(T_d)^2\le 1$ and Lemma~\ref{lem:densestat}. Hence
\[
\E\!\left[s_{d,c}(T_d)^2\frac{T_d}{d}\right]\to a_\rho^2\tau.
\]
Similarly,
\[
\E\!\left|s_{d,c}(T_d)\frac{\theta_d^\top X_d}{d}-a_\rho \rho\right|
\le
\E\!\left|\frac{\theta_d^\top X_d}{d}-\rho\right|
+
\rho\,\E\!\left|s_{d,c}(T_d)-a_\rho\right|
\to 0,
\]
so
\[
\E\!\left[s_{d,c}(T_d)\frac{\theta_d^\top X_d}{d}\right]\to a_\rho\rho.
\]
Finally, $\norm{\theta_d}^2/d\to \rho$ by assumption. Therefore
\[
\frac{\risk(\theta_d,\delta_{d,c})}{d}
\to
a_\rho^2(1+\rho)-2a_\rho\rho+\rho
=
\frac{\rho}{1+\rho},
\]
which is exactly the constant-shrinkage risk with $a_\rho=\rho/(1+\rho)$.
\end{proof}

\begin{proof}[\textbf{Proof of Theorem~\ref{thm:dense2}}]
Write
\[
\nu_d=\norm{\theta_d}^2=\rho d+\kappa+o(1),
\quad
\tau_0:=1+\rho,
\quad
\tau_d:=\frac{T_d}{d}.
\]
By Lemma~\ref{lem:riskid},
\[
\risk(\theta_d,\delta_{d,c})
=
\nu_d+2c-\E\!\left[T_d s_{d,c}(T_d)^2\right].
\]
Choose a compact interval $K=[\tau_-,\tau_+]\subset(1,\infty)$ with $\tau_0\in(\tau_-,\tau_+)$. Since
\[
T_d-(d+\nu_d)=\sum_{j=1}^d
\left\{\varepsilon_{j,d}^2-1+2\theta_{j,d}\varepsilon_{j,d}\right\},
\quad \varepsilon_d\sim N_d(0,I_d),
\]
is a sum of independent centered variables and \(\nu_d=O(d)\), Rosenthal's inequality gives, for every fixed integer \(m\ge 1\),
\[
\E\!\left|T_d-(d+\nu_d)\right|^{2m}=O(d^m).
\]
Indeed, if \(Z_{j,d}:=\varepsilon_{j,d}^2-1+2\theta_{j,d}\varepsilon_{j,d}\), then
\[
\sum_{j=1}^d \E Z_{j,d}^2=2d+4\nu_d=O(d),
\quad
\sum_{j=1}^d \E |Z_{j,d}|^{2m}
\le C_m\left(d+\sum_{j=1}^d |\theta_{j,d}|^{2m}\right)
\le C_m(d+\nu_d^m)=O(d^m).
\]
Thus, with \(\mu_d:=\E\tau_d=1+\nu_d/d=\tau_0+\kappa/d+o(d^{-1})\),
\[
\E\!\left|\tau_d-\mu_d\right|^{2m}=O(d^{-m})
\]
Because \(\mu_d\to\tau_0\) and \(\tau_0\in(\tau_-,\tau_+)\), there is \(\eta>0\) such that, for all large \(d\),
\[
\{\tau_d\notin K\}\subseteq \{|\tau_d-\mu_d|\ge \eta\}.
\]
Taking \(m=2\) in the preceding moment bound gives
\[
\Pbb(\tau_d\notin K)=O(d^{-2}),
\quad
\E\!\left[|\tau_d-\mu_d|\mathbf 1_{\{\tau_d\notin K\}}\right]=O(d^{-2}),
\]
where the second bound uses
\(|x|\mathbf 1_{\{|x|\ge\eta\}}\le \eta^{-3}|x|^4\). Since \(\mu_d\) is bounded,
\[
\E\!\left[\tau_d\mathbf 1_{\{\tau_d\notin K\}}\right]
\le
\mu_d\Pbb(\tau_d\notin K)
+
\E\!\left[|\tau_d-\mu_d|\mathbf 1_{\{\tau_d\notin K\}}\right]
=O(d^{-2}).
\]
Consequently, since \(0\le s_{d,c}\le 1\) and \(T_d=d\tau_d\),
\[
\E\!\left[T_d s_{d,c}(T_d)^2\mathbf 1_{\{\tau_d\notin K\}}\right]
\le
d\,\E\!\left[\tau_d\mathbf 1_{\{\tau_d\notin K\}}\right]
=o(1).
\]
Therefore
\[
\E[T_d s_{d,c}(T_d)^2]
=
\E\!\left[T_d s_{d,c}(T_d)^2\mathbf 1_{\{\tau_d\in K\}}\right]
+o(1).
\]

Set
\[
f(\tau):=\tau a(\tau)^2=\tau-2+\tau^{-1},
\quad
g_c(\tau):=2\tau a(\tau)b_c(\tau).
\]
For the tail estimates below, use the algebraically equivalent extension
\[
g_c(\tau)=2(c-2)+\frac{4(\tau-1)}{\tau},
\quad \tau>0,
\]
which agrees with \(2\tau a(\tau)b_c(\tau)\) on \((1,\infty)\).
On $\{\tau_d\in K\}$, Lemma~\ref{lem:denseLaplace} implies
\[
T_d s_{d,c}(T_d)^2
=
d\,f(\tau_d)+g_c(\tau_d)+o(1)
\]
uniformly in probability and hence in $L^1$. Consequently,
\[
\E[T_d s_{d,c}(T_d)^2]
=
d\,\E\!\left[f(\tau_d)\mathbf 1_{\{\tau_d\in K\}}\right]
+
\E\!\left[g_c(\tau_d)\mathbf 1_{\{\tau_d\in K\}}\right]
+
o(1).
\]

We next remove the indicators in the two deterministic expansion terms. Write
\[
A_d^+:=\{\tau_d>\tau_+\},
\quad
A_d^-:=\{\tau_d<\tau_-\}.
\]
On \(A_d^+\), we have \(\tau_d>1\), and hence
\[
0\le f(\tau_d)=\frac{(\tau_d-1)^2}{\tau_d}\le \tau_d,
\quad
|g_c(\tau_d)|\le C_{c,K}.
\]
The preceding moment argument gives
\[
d\,\E\!\left[f(\tau_d)\mathbf 1_{A_d^+}\right]
\le d\,\E\!\left[\tau_d\mathbf 1_{A_d^+}\right]=o(1),
\quad
\E\!\left[|g_c(\tau_d)|\mathbf 1_{A_d^+}\right]=o(1).
\]
It remains only to control the inverse powers on the lower tail \(A_d^-\). On \((0,\tau_-)\),
\[
f(\tau)\le C_K+\tau^{-1},
\quad
|g_c(\tau)|\le C_{c,K}(1+\tau^{-1}).
\]
Let \(T_d\sim\chi^2_d(\nu_d)\). We claim that
\[
\E\!\left[\tau_d^{-1}\mathbf 1_{A_d^-}\right]
=O(e^{-\eta_0 d})
\]
for some \(\eta_0>0\). To prove this, use the Poisson-mixture representation
\(T_d\mid N\sim\chi^2_{d+2N}\), \(N\sim\mathrm{Poisson}(\nu_d/2)\). If \(W_m\sim\chi^2_m\), then, for \(m>2\),
\[
\E\!\left[W_m^{-1}\mathbf 1_{\{W_m\le x\}}\right]
=\frac{1}{m-2}\Pbb(W_{m-2}\le x).
\]
Therefore, for \(x>0\) and \(d>2\),
\[
\E\!\left[T_d^{-1}\mathbf 1_{\{T_d\le x\}}\right]
\le
\frac{1}{d-2}\Pbb(V_d\le x),
\quad
V_d\sim\chi^2_{d-2}(\nu_d).
\]
For any fixed \(u>0\), Chernoff's bound and the Laplace transform of a noncentral chi-square variable give
\[
\Pbb(V_d\le d\tau_-)
\le
\exp\!\left\{
u d\tau_-
-\frac{d-2}{2}\log(1+2u)
-\frac{\nu_d u}{1+2u}
\right\}.
\]
Because \(\tau_-<\tau_0=1+\rho\) and \(\nu_d/d\to\rho\), the exponent divided by \(d\) has derivative
\(\tau_--(1+\rho)<0\) at \(u=0\). Hence we may choose \(u>0\) small enough that, for all large \(d\),
\[
\Pbb(V_d\le d\tau_-)
\le e^{-\eta_0 d}
\]
for some \(\eta_0>0\). Taking \(x=d\tau_-\) in the inverse-moment bound yields
\[
\E\!\left[\tau_d^{-1}\mathbf 1_{A_d^-}\right]
=d\,\E\!\left[T_d^{-1}\mathbf 1_{\{T_d\le d\tau_-\}}\right]
\le \frac{d}{d-2}e^{-\eta_0d}
=O(e^{-\eta_0d}).
\]
Combining this inverse-tail estimate with \(\Pbb(A_d^-)=O(d^{-2})\), we obtain
\[
d\,\E\!\left[|f(\tau_d)|\mathbf 1_{A_d^-}\right]=o(1),
\quad
\E\!\left[|g_c(\tau_d)|\mathbf 1_{A_d^-}\right]=o(1).
\]
Together with the upper-tail bounds, this proves
\[
d\,\E\!\left[|f(\tau_d)|\mathbf 1_{\{\tau_d\notin K\}}\right]
+
\E\!\left[|g_c(\tau_d)|\mathbf 1_{\{\tau_d\notin K\}}\right]
=o(1).
\]

Next,
\[
\Var(\tau_d)=\frac{2d+4\nu_d}{d^2}
=
\frac{2(1+2\rho)}{d}+o(d^{-1}),
\quad
\E[|\tau_d-\tau_0|^3]=O(d^{-3/2}).
\]
The preceding tail estimate gives
\(\E[|f(\tau_d)|\mathbf 1_{\{\tau_d\notin K\}}]=o(d^{-1})\). On \(K\), the third derivative of \(f\) is bounded, so Taylor's theorem gives
\[
f(\tau_d)
=
f(\tau_0)
+
f'(\tau_0)(\tau_d-\tau_0)
+
\frac12 f''(\tau_0)(\tau_d-\tau_0)^2
+
R_d,
\quad
|R_d|\le C_K|\tau_d-\tau_0|^3,
\]
on \(\{\tau_d\in K\}\). The same moment bound also shows that, for \(j=0,1,2\), replacing
\(\E[(\tau_d-\tau_0)^j\mathbf 1_{\{\tau_d\in K\}}]\) by \(\E[(\tau_d-\tau_0)^j]\) changes the expression by \(O(d^{-2})\). Therefore
\[
\E[f(\tau_d)]
=
f(\tau_0)
+
f'(\tau_0)\E[\tau_d-\tau_0]
+
\frac12 f''(\tau_0)\E[(\tau_d-\tau_0)^2]
+
o(d^{-1}),
\]
and hence
\[
d\,\E[f(\tau_d)]
=
d\,f(\tau_0)
+
\kappa f'(\tau_0)
+
(1+2\rho)f''(\tau_0)
+
o(1).
\]
Similarly, the outside-\(K\) contribution to \(\E[g_c(\tau_d)]\) is \(o(1)\), while \(g_c\) is Lipschitz on \(K\) and \(\tau_d\to\tau_0\) in \(L^1\). Thus
\[
\E[g_c(\tau_d)] = g_c(\tau_0)+o(1).
\]

Therefore
\[
\E[T_d s_{d,c}(T_d)^2]
=
d\,f(\tau_0)
+
\kappa f'(\tau_0)
+
(1+2\rho)f''(\tau_0)
+
g_c(\tau_0)
+
o(1).
\]

Now
\[
f(\tau_0)=\tau_0-2+\tau_0^{-1},
\quad
f'(\tau_0)=1-\tau_0^{-2},
\quad
f''(\tau_0)=2\tau_0^{-3},
\]
and
\[
g_c(\tau)=2\tau\left(1-\frac1\tau\right)\left(\frac{2}{\tau}+\frac{c-2}{\tau-1}\right)
=
2(c-2)+\frac{4(\tau-1)}{\tau}.
\]
Substituting these identities into the exact risk formula gives
\begin{align*}
\risk(\theta_d,\delta_{d,c})
&=
\rho d+\kappa+2c
-
d\,f(\tau_0)
-
\kappa f'(\tau_0)
-
(1+2\rho)f''(\tau_0)
-
g_c(\tau_0)
+
o(1)\\
&=
d\frac{\rho}{1+\rho}
+
\frac{\kappa}{(1+\rho)^2}
+
\frac{4}{1+\rho}
-
\frac{2(1+2\rho)}{(1+\rho)^3}
+
o(1).
\end{align*}
Since
\[
\frac{4}{1+\rho}
-
\frac{2(1+2\rho)}{(1+\rho)^3}
=
\frac{2(1+2\rho+2\rho^2)}{(1+\rho)^3},
\]
this is exactly the claimed expansion.
\end{proof}

\subsection{Results in Sections~\ref{subsec:properdefaults}--\ref{subsec:globalpriors}}
\label{app:properproofs}
\subsubsection{Key ingredients for Proposition~\ref{prop:proper}}

As in Section~\ref{sec:model}, the posterior mean under the properized hierarchy is an isotropic shrinkage rule. Writing $Y=g/(1+g)$ as before,
\[
\delta_{d,c,\ell}(x)=s_{d,c,\ell}(\norm{x}^2)x,
\qquad
s_{d,c,\ell}(t):=\E_\ell[Y\mid T_d=t],
\]
where $\E_\ell$ denotes expectation under the mixing density
$p_\ell(g)\propto g^{c/2-1}\ell(g)$. After rotating coordinates so that
$\theta_d=(r_d,0,\ldots,0)$, every isotropic rule
$\widehat\theta_d=s(T_d)X_d$ with $\E[s(T_d)^2T_d]<\infty$ satisfies
\begin{align}
\risk(\theta_d,\widehat\theta_d)-\norm{\theta_d}^2
&=\E[s(T_d)^2\norm{X_d}^2]
   -2\E[s(T_d)\theta_d^\top X_d]\notag\\
&=\E[s(T_d)^2T_d]-2r_d\E[s(T_d)X_{1,d}].
\label{eq:proper-risk-decomp}
\end{align}
The comparison lemma below controls the properized shrinkage factor by the pure-power rule, and the boundary lemma identifies its local limit. We use these two facts separately for the norm and cross terms in \eqref{eq:proper-risk-decomp}.

\begin{lemma}[Monotone properization lowers shrinkage]
\label{lem:propercompare}
For every $d>c$, every $t\ge0$, and every bounded nonnegative nonincreasing function
$\ell:[0,\infty)\to[0,\infty)$ that is positive on a neighborhood of zero,
\[
0\le s_{d,c,\ell}(t)\le s_{d,c}(t).
\]
\end{lemma}

\begin{proof}[\textbf{Proof of Lemma~\ref{lem:propercompare}}]
Fix $t\ge0$. Let $P_{d,c,t}$ denote the posterior law of $Y$ given $T_d=t$ under the pure-power prior $\pi_{d,c}$. By Lemma~\ref{prop:beta},
\[
s_{d,c,\ell}(t)
=
\frac{\E_{P_{d,c,t}}\!\left[Y\,W(Y)\right]}
{\E_{P_{d,c,t}}\!\left[W(Y)\right]},
\qquad
W(y):=\ell\!\left(\frac{y}{1-y}\right).
\]
The denominator is strictly positive because $\ell$ is positive near zero. Since
$y\mapsto y/(1-y)$ is increasing and $\ell$ is nonincreasing, $W$ is nonincreasing on $(0,1)$. If $Y'$ is an independent copy of $Y$ under $P_{d,c,t}$, then
$$
\operatorname{Cov}_{P_{d,c,t}}(Y,W(Y)) =\frac12\E\!\left[(Y-Y')\{W(Y)-W(Y')\}\right] \le0,
$$
because the two factors in the product have opposite signs. Hence
\[
\E_{P_{d,c,t}}[Y W(Y)]
\le
\E_{P_{d,c,t}}[Y]\,\E_{P_{d,c,t}}[W(Y)].
\]
Dividing by the positive denominator gives
$s_{d,c,\ell}(t)\le s_{d,c}(t)$. Nonnegativity is immediate.
\end{proof}

\begin{lemma}[Robustness of the local shrinkage limit under bounded monotone properizations]
\label{lem:properboundary}
Fix $c>0$ and let $h_c$ be the profile from Theorem~\ref{thm:main}. For every $M<\infty$,
\[
\sup_{|z|\le M}\left|
\sqrt d\, s_{d,c,\ell}(d+\sqrt d\,z)-h_c(z)
\right|\to 0.
\]
\end{lemma}

\begin{proof}[\textbf{Proof of Lemma~\ref{lem:properboundary}}]
Put $b=c/2$ and, for $m=0,1$, define
\[
J^{(\ell)}_{d,m}(z)
:=
\int_0^{\sqrt d}
 u^{b-1+m}
\left(1-\frac{u}{\sqrt d}\right)^{(d-c)/2-1}
\exp\!\left(\frac{(d+\sqrt d\,z)u}{2\sqrt d}\right)
\ell\!\left(\frac{u}{\sqrt d-u}\right)du.
\]
The change of variables $u=\sqrt d\,y$ in the posterior integrals gives
\[
\sqrt d\,s_{d,c,\ell}(d+\sqrt d\,z)
=
\frac{J^{(\ell)}_{d,1}(z)}{J^{(\ell)}_{d,0}(z)}.
\]
Fix $M<\infty$. For fixed $U<\infty$, Taylor's formula for
$\log(1-u/\sqrt d)$ yields
$$
\sup_{\substack{0\le u\le U\\ |z|\le M}}
\left|
\left(\frac{d-c}{2}-1\right)\log\left(1-\frac{u}{\sqrt d}\right)
+\frac{(d+\sqrt d\,z)u}{2\sqrt d}
-\left(\frac{zu}{2}-\frac{u^2}{4}\right)
\right| \longrightarrow0.
$$
Continuity of $\ell$ at zero also gives
\[
\sup_{0\le u\le U}
\left|\ell\!\left(\frac{u}{\sqrt d-u}\right)-\ell(0)\right|
\longrightarrow0.
\]
Thus the integrands converge uniformly on $[0,U]\times[-M,M]$ to
\[
\ell(0)u^{b-1+m}\exp\!\left(\frac{zu}{2}-\frac{u^2}{4}\right).
\]

To control the tails, use $\log(1-x)\le-x-x^2/2$. For all sufficiently large
$d$, all $0<u<\sqrt d$, and all $|z|\le M$,
$$
\left(\frac{d-c}{2}-1\right)\log\left(1-\frac{u}{\sqrt d}\right)
+\frac{(d+\sqrt d\,z)u}{2\sqrt d} \le
\frac{Mu}{2}+u-\frac{u^2}{8}.
$$
Since $0\le\ell(g)\le\ell(0)$, the finite-$d$ integrands, extended by zero
for $u\ge\sqrt d$, are dominated by a constant multiple of
\[
u^{b-1+m}\exp\!\left(\frac{Mu}{2}+u-\frac{u^2}{8}\right),
\]
which is integrable for $m=0,1$. Dominated convergence therefore gives
\[
\sup_{|z|\le M}
\left|J^{(\ell)}_{d,m}(z)-\ell(0)I_{b+m}(z)\right|
\longrightarrow0,
\qquad m=0,1,
\]
where $I_b(z)=\int_0^\infty
u^{b-1}\exp\!\left(\frac{zu}{2}-\frac{u^2}{4}\right)du$ as defined in Appendix \ref{app:weakproofs}.
Because $I_b$ is continuous and strictly positive,
$\ell(0)\inf_{|z|\le M}I_b(z)>0$; hence $J^{(\ell)}_{d,0}$ is bounded away
from zero uniformly on $[-M,M]$ for all large $d$. Taking ratios and using
$I_{b+1}(z)/I_b(z)=h_c(z)$ proves the result.
\end{proof}

\begin{proof}[\textbf{Proof of Proposition~\ref{prop:proper}}]
It is enough to prove the critical-regime assertion. Write
\[
\nu_d:=\norm{\theta_d}^2,
\qquad
\beta_d:=\frac{\nu_d}{\sqrt d}\longrightarrow\beta,
\qquad
Z_d:=\frac{T_d-d}{\sqrt d},
\]
and rotate coordinates so that $\theta_d=(r_d,0,\ldots,0)$, where
$r_d^2=\nu_d$. Define
\[
\widetilde v_d(z):=\sqrt d\,s_{d,c,\ell}(d+\sqrt d\,z),
\qquad
\widetilde V_d:=\widetilde v_d(Z_d).
\]
By Lemma~\ref{lem:properboundary}, $\widetilde v_d\to h_c$ locally uniformly,
and Lemma~\ref{lem:Zd} gives $Z_d\Rightarrow Z_\beta$. For every $M<\infty$
and $\varepsilon>0$,
$$
\Pbb\left(|\widetilde v_d(Z_d)-h_c(Z_d)|>\varepsilon\right)
\le \Pbb(|Z_d|>M) +
\mathbf 1\left\{
\sup_{|z|\le M}|\widetilde v_d(z)-h_c(z)|>\varepsilon
\right\}.
$$
Tightness of $\{Z_d\}$ and local uniform convergence make the right-hand
side arbitrarily small. Since $h_c$ is continuous, Slutsky's lemma yields
\begin{equation}
\label{eq:proper-V-dist}
\widetilde V_d\Rightarrow h_c(Z_\beta).
\end{equation}

We next obtain the domination needed for expectations. Let
$v_d(z)=\sqrt d\,s_{d,c}(d+\sqrt d\,z)$ denote the pure-power quantity. We
claim that, for all sufficiently large $d$,
\begin{equation}
\label{eq:pure-power-envelope}
v_d(z)\le K_c(1+z_+),
\qquad z\ge-\sqrt d,
\end{equation}
for a constant $K_c$ independent of $d$ and $z$.
Indeed, put $t=d+\sqrt d\,z$, $\mu_d=s_{d,c}(t)$,
$v=\sqrt d\,\mu_d$, and $\tau=t/d$. Proposition~\ref{prop:moment} and
$q_d(t)\ge\mu_d^2$ imply
\[
\tau v^2\le c+zv.
\]
If $z>1$, then $\tau\ge1$ and
\[
v\le\frac{z+\sqrt{z^2+4c}}2\le z+\sqrt c.
\]
If $z<-1$, the identity $tq_d(t)=c+(t-d)\mu_d\ge0$ gives
$v\le c/(-z)\le c$. If $|z|\le1$, Lemma~\ref{lem:boundary} bounds $v_d$
uniformly for all large $d$. These three cases prove
\eqref{eq:pure-power-envelope}.

Lemma~\ref{lem:propercompare} now gives
$0\le\widetilde V_d\le v_d(Z_d)$. Therefore, by
\eqref{eq:pure-power-envelope} and Lemma~\ref{lem:Zd},
\[
\sup_d\E\left[v_d(Z_d)^2\right]
\le K_c^2\sup_d\E\left[(1+Z_d^+)^2\right]<\infty.
\]
Thus $\{\widetilde V_d\}$ is uniformly integrable, and
\begin{equation}
\label{eq:proper-V-mean}
\E\widetilde V_d\longrightarrow\E h_c(Z_\beta).
\end{equation}

For the norm term, define
\[
\widetilde G_d(z)
:=\left(1+\frac{z}{\sqrt d}\right)\widetilde v_d(z)^2,
\qquad z\ge-\sqrt d.
\]
The local uniform convergence of $\widetilde v_d$ implies
$\widetilde G_d\to h_c^2$ locally uniformly. Hence, for every $M$ and
$\varepsilon>0$,
$$
\Pbb\left(|\widetilde G_d(Z_d)-h_c(Z_d)^2|>\varepsilon\right)
\le \Pbb(|Z_d|>M) +
\mathbf 1\left\{
\sup_{|z|\le M}|\widetilde G_d(z)-h_c(z)^2|>\varepsilon
\right\},
$$
so
\[
T_d s_{d,c,\ell}(T_d)^2
=\widetilde G_d(Z_d)
\Rightarrow h_c(Z_\beta)^2.
\]
Moreover, Lemma~\ref{lem:propercompare} and Proposition~\ref{prop:moment}
give
\[
0\le T_d s_{d,c,\ell}(T_d)^2
\le C_c\{1+(Z_d^+)^2\}.
\]
The right-hand side is bounded in $L^2$ by Lemma~\ref{lem:Zd}; consequently,
\begin{equation}
\label{eq:proper-norm-limit}
\E[T_d s_{d,c,\ell}(T_d)^2]
\longrightarrow
\E[h_c(Z_\beta)^2].
\end{equation}

It remains to control the cross term in \eqref{eq:proper-risk-decomp}. Write
$X_{1,d}=r_d+\xi_d$, where $\xi_d\sim N(0,1)$. Since
$s_{d,c,\ell}(T_d)=\widetilde V_d/\sqrt d$, \eqref{eq:proper-V-mean} implies
\[
r_d^2\E[s_{d,c,\ell}(T_d)]
=\beta_d\E\widetilde V_d
\longrightarrow
\beta\E h_c(Z_\beta).
\]
Let
$W_d=\sum_{j=2}^dX_{j,d}^2$. Then $W_d\sim\chi^2_{d-1}$, it is independent
of $\xi_d$, and $T_d\ge W_d$. Thus, for $d>3$,
\[
\E\left[\frac{\xi_d^2}{T_d}\right]
\le \E[\xi_d^2]\ \E[W_d^{-1}]
=\frac1{d-3}.
\]
By Cauchy--Schwarz and \eqref{eq:proper-norm-limit},
\begin{align*}
\left|r_d\E[s_{d,c,\ell}(T_d)\xi_d]\right|
&\le r_d
\{\E[T_d s_{d,c,\ell}(T_d)^2]\}^{1/2}
\left\{\E\left[\frac{\xi_d^2}{T_d}\right]\right\}^{1/2}\\
&=O(d^{1/4})O(1)O(d^{-1/2})=o(1).
\end{align*}
Therefore
\[
r_d\E[s_{d,c,\ell}(T_d)X_{1,d}]
\longrightarrow
\beta\E h_c(Z_\beta).
\]
Substitution into \eqref{eq:proper-risk-decomp}, together with
\eqref{eq:proper-norm-limit}, gives
\begin{align*}
\risk(\theta_d,\delta_{d,c,\ell})-\norm{\theta_d}^2
&=\E[T_d s_{d,c,\ell}(T_d)^2]
 -2r_d\E[s_{d,c,\ell}(T_d)X_{1,d}]\\
&\longrightarrow
\E[h_c(Z_\beta)^2]-2\beta\E h_c(Z_\beta)
=L_c(\beta).
\end{align*}
If $\norm{\theta_d}^2\to\nu<\infty$, then $\beta=0$ and
Lemma~\ref{lem:profile}(iii) gives $L_c(0)=c$; hence
$\risk(\theta_d,\delta_{d,c,\ell})\to\nu+c$.

For the equivalent SD-scale formulation, set $g=\tau^2$ and
$\ell(g)=m(\sqrt g)$. Since $d\tau/dg=(2\sqrt g)^{-1}$,
\[
p_g(g)=\frac{p_\tau(\sqrt g)}{2\sqrt g}
\propto g^{c/2-1}\ell(g).
\]
The assumptions on $m$ imply the corresponding assumptions on $\ell$, so
the variance-scale result applies.
\end{proof}

\subsubsection{Single-global-scale transfer and change-of-variables formulas}
\label{app:hyperpriors}

This subsection collects the auxiliary change-of-variables formulas used in Section~\ref{subsec:globalpriors}. Because the SD-scale formulation is already built into Proposition~\ref{prop:proper}, it remains only to verify the near-zero exponents and regularity conditions for the named single global-scale hyperpriors.

\begin{proposition}[Auxiliary change-of-variables dictionary for single global-scale hyperpriors]
\label{prop:commonhyperapp}
For the hierarchy $\theta\mid \tau\sim N_d(0,\tau^2 I_d)$, the following statements hold.
\begin{enumerate}
\item[(i)] If $\tau$ has a half-$t_{\nu_0}(s)$, half-Cauchy$(s)$, half-normal$(s)$,
exponential$(b)$, or truncated-flat density on $[0,A]$, then $p(\tau)=K+o(1)$ as
$\tau\downarrow 0$ for some $K\in(0,\infty)$, so the local index is $c=1$.

\item[(ii)] If $g=\tau^2\sim \mathrm{Gamma}(a,b)$, then
\[
p(\tau)=\frac{2b^a}{\Gamma(a)}\tau^{2a-1}e^{-b\tau^2},
\quad \tau>0,
\]
so the local index is $c=2a$.

\item[(iii)] If $W=w_0\tau^2$ for some fixed $w_0>0$ and $W\sim\mathrm{Beta\text{-}prime}(a,b)$,
equivalently $R^2=W/(1+W)\sim\mathrm{Beta}(a,b)$, then
\[
p(\tau)=\frac{2w_0^a}{\mathrm{B}(a,b)}\tau^{2a-1}(1+w_0\tau^2)^{-(a+b)},
\quad \tau>0,
\]
so the local index is $c=2a$. In particular, $a=\tfrac12$ is the SD-scale class, whereas $a=1$
is the variance-flat class.

\end{enumerate}
In cases (i)--(iii), the density admits the representation $p(\tau)=K\tau^{c-1}m(\tau)$ with $m$ bounded, nonnegative, nonincreasing, continuous at zero, and $m(0)=1$. Consequently, Proposition~\ref{prop:proper} applies with the indicated exponent $c$.
\end{proposition}

\begin{proof}[\textbf{Proof of Proposition~\ref{prop:commonhyperapp}}]
For part~(i), it is enough to evaluate the density at the origin. A half-$t_{\nu_0}(s)$ density equals
$2s^{-1}t_{\nu_0}(\tau/s)$ for $\tau>0$, a half-Cauchy is the case
$\nu_0=1$, and a half-normal density equals
$2s^{-1}\phi(\tau/s)$. Each is continuous at zero with a finite positive value. The exponential density has value $b$ at zero, and the truncated-flat density has value $1/A$ on $[0,A]$. Hence in every case
$p(\tau)=p(0)+o(1)$ with $p(0)\in(0,\infty)$, which is the local form
$K\tau^{c-1}$ with $c=1$.

For part~(ii), if $g\sim\mathrm{Gamma}(a,b)$ with rate $b$, then
\[
p_g(g)=\frac{b^a}{\Gamma(a)}g^{a-1}e^{-bg},
\qquad g>0.
\]
Under $g=\tau^2$, the Jacobian is $dg/d\tau=2\tau$, so
\[
p_\tau(\tau)
=p_g(\tau^2)2\tau
=
\frac{2b^a}{\Gamma(a)}\tau^{2a-1}e^{-b\tau^2}.
\]
The factor $e^{-b\tau^2}$ tends to one at zero; therefore the local index is
$c=2a$.

For part~(iii), write $U=R^2$. The transformation
$W=U/(1-U)$ has inverse $U=W/(1+W)$ and derivative
$dU/dW=(1+W)^{-2}$. Thus
\begin{align*}
p_W(w)
&=
\frac1{\mathrm B(a,b)}
\left(\frac{w}{1+w}\right)^{a-1}
\left(\frac1{1+w}\right)^{b-1}
\frac1{(1+w)^2}\\
&=
\frac1{\mathrm B(a,b)}w^{a-1}(1+w)^{-(a+b)},
\qquad w>0.
\end{align*}
Now set $w=w_0\tau^2$, whose Jacobian is $2w_0\tau$. Then
\begin{align*}
p_\tau(\tau)
&=p_W(w_0\tau^2)2w_0\tau\\
&=
\frac{2w_0^a}{\mathrm B(a,b)}
\tau^{2a-1}(1+w_0\tau^2)^{-(a+b)}.
\end{align*}
The last factor tends to one at zero, so again $c=2a$.

It remains to check the regularity conditions used by Proposition~\ref{prop:proper}. In part~(i), after division by $p_\tau(0)$, the residual factor is bounded, nonnegative, nonincreasing, continuous at zero, and equal to one there for each named density. In parts~(ii) and~(iii), the corresponding factors are $m(\tau)=e^{-b\tau^2}$ and $m(\tau)=(1+w_0\tau^2)^{-(a+b)}$, respectively; both have the same properties. Integrability follows because the displayed expressions are probability densities. Thus Proposition~\ref{prop:proper} applies with $c=1$ in part~(i) and $c=2a$ in parts~(ii)--(iii), completing the single-global-scale dictionary.
\end{proof}

\subsection{Section~\ref{subsec:beyondglobal}: Beyond a single global scale}
\label{app:beyond-details}

\subsubsection{Proof of Proposition~\ref{prop:bounded-mult-transfer}: bounded coordinate-multiplier transfer}
\label{app:coord-mult-proof}

Throughout, let $X_{dj}=\theta_{dj}+\varepsilon_{dj}$ with independent
$\varepsilon_{dj}\sim N(0,1)$, and write
\[
T_d=\norm{X_d}^2,
\qquad
Z_d=\frac{T_d-d}{\sqrt d}.
\]
Under either signal regime in the proposition,
$Z_d\Rightarrow Z_\beta\sim N(\beta,2)$, where
$\beta=\lim \norm{\theta_d}^2/\sqrt d$. Moreover,
\[
Z_d
=
\frac1{\sqrt d}\sum_{j=1}^d(\varepsilon_{dj}^2-1)
+
\frac{2}{\sqrt d}\theta_d^\top\varepsilon_d
+
\frac{\norm{\theta_d}^2}{\sqrt d}.
\]
For every fixed $k\ge2$, Rosenthal's inequality gives
\[
\sup_d\E\left|
\frac1{\sqrt d}\sum_{j=1}^d(\varepsilon_{dj}^2-1)
\right|^k<\infty.
\]
The second term is Gaussian with variance
$4\norm{\theta_d}^2/d=O(d^{-1/2})$, and the last term is bounded. Hence
\begin{equation}
\label{eq:app-Zmom-gl}
\sup_d\E|Z_d|^k<\infty
\qquad\text{for every fixed }k<\infty.
\end{equation}

For fixed $g$, define the one-coordinate marginal density ratio
\[
r_g(x)=\E_A\left[(1+gA)^{-1/2}
\exp\left\{\frac{x^2gA}{2(1+gA)}\right\}\right]
\]
and the corresponding conditional shrinkage factor
\[
q_g(x)=
\frac{\E_A\left[\frac{gA}{1+gA}(1+gA)^{-1/2}
\exp\left\{\frac{x^2gA}{2(1+gA)}\right\}\right]}
{r_g(x)}.
\]
Normal conjugacy gives
\begin{equation}
\label{eq:app-exact-postmean-gl}
\delta_{dj}(X_d)=X_{dj}\E[q_g(X_{dj})\mid X_d],
\end{equation}
where the posterior density of $g$ is proportional to
\[
g^{c/2-1}\ell(g)\mathbf 1_{(0,G)}(g)
\prod_{j=1}^d r_g(X_{dj}).
\]
The next two lemmas identify the local boundary experiment and exclude
posterior mass at scales bounded away from zero. The posterior localization, radial approximation, and risk calculation are
carried out directly in the proof of the proposition.

\begin{lemma}[Local likelihood expansion]
\label{lem:mult-local-likelihood}
Let
\[
\alpha_1(x)=\frac{x^2-1}{2},
\qquad
\alpha_2(x)=\frac{x^4-6x^2+3}{8},
\qquad
\mu_2=\E[A^2].
\]
There are constants $\eta_0>0$ and $C<\infty$ such that, uniformly in $x$
and $0\le g\le\eta_0$,
\begin{align}
\log r_g(x)
&=\mu\alpha_1(x)g+B_2(x)g^2+\rho_g(x),
\label{eq:app-logr-gl}\\
B_2(x)
&=\mu_2\alpha_2(x)-\frac{\mu^2}{2}\alpha_1(x)^2,
\nonumber\\
|\rho_g(x)|
&\le Cg^3(1+|x|^6)e^{Cgx^2},
\label{eq:app-rho-gl}\\
q_g(x)
&=\mu g+O\{g^2(1+x^2)e^{Cgx^2}\}.
\label{eq:app-qg-exp}
\end{align}
Moreover, if
\[
L_d(u)=\sum_{j=1}^d
\log r_{u/(\mu\sqrt d)}(X_{dj}),
\qquad
0\le u\le\eta_0\mu\sqrt d,
\]
then, for every fixed $M<\infty$,
\begin{equation}
\label{eq:app-boundary-gl}
\sup_{0\le u\le M}
\left|L_d(u)-\left(\frac{u}{2}Z_d-\frac{u^2}{4}\right)\right|
\to0
\quad\text{in probability}.
\end{equation}
\end{lemma}

\begin{proof}[\textbf{Proof of Lemma~\ref{lem:mult-local-likelihood}}]
Set
\[
\psi_g(a,x):=(1+ga)^{-1/2}
\exp\left\{\frac{x^2ga}{2(1+ga)}\right\}.
\]
Because $0\le a\le A_+$, Taylor's theorem on a sufficiently small interval
$0\le g\le\eta_0$ gives, uniformly in $a$ and $x$,
\begin{equation}
\label{eq:app-psi-exp-gl}
\psi_g(a,x)
=1+a\alpha_1(x)g+a^2\alpha_2(x)g^2
+O\{g^3(1+|x|^6)e^{Cgx^2}\}.
\end{equation}
Taking expectation over $A$ yields
\[
r_g(x)=1+\mu\alpha_1(x)g+\mu_2\alpha_2(x)g^2
+O\{g^3(1+|x|^6)e^{Cgx^2}\}.
\]
Since $r_g(x)$ is bounded below by $(1+\eta_0A_+)^{-1/2}$, expanding
$\log r_g(x)$ gives \eqref{eq:app-logr-gl}--\eqref{eq:app-rho-gl}. Likewise,
\[
\E_A\left[\frac{gA}{1+gA}\psi_g(A,x)\right]
=\mu g+O\{g^2(1+x^2)e^{Cgx^2}\},
\]
and division by $r_g(x)$ proves \eqref{eq:app-qg-exp}.

It remains to sum the second-order and remainder terms. If
$X\sim N(\theta,1)$, then
\[
\E[X^4-6X^2+3]=\theta^4,
\qquad
\E[(X^2-1)^2]=\theta^4+4\theta^2+2.
\]
Therefore
\[
\frac1d\sum_{j=1}^d\E B_2(X_{dj})
=
\frac{\mu_2-\mu^2}{8d}\sum_j\theta_{dj}^4
-
\frac{\mu^2}{2d}\sum_j\theta_{dj}^2
-
\frac{\mu^2}{4}
\longrightarrow-\frac{\mu^2}{4}.
\]
Because $B_2$ has degree four and the means $\theta_{dj}$ are uniformly
bounded,
\[
\Var\left(\frac1d\sum_{j=1}^dB_2(X_{dj})\right)
=\frac1{d^2}\sum_{j=1}^d\Var\{B_2(X_{dj})\}
=O(d^{-1}).
\]
Consequently,
\begin{equation}
\label{eq:app-B2avg-gl}
\overline B_d:=\frac1d\sum_{j=1}^dB_2(X_{dj})
\longrightarrow-\frac{\mu^2}{4}
\quad\text{in }L^2.
\end{equation}
After decreasing $\eta_0$ if necessary, Gaussian exponential moments and
the bounded-coordinate assumption give
\[
\sup_d\E H_d<\infty,
\qquad
H_d:=\frac1d\sum_{j=1}^d
(1+|X_{dj}|^6)e^{C\eta_0X_{dj}^2},
\]
so
\begin{equation}
\label{eq:app-expavg-gl}
H_d=O_{\Pbb}(1).
\end{equation}
Substituting $g=u/(\mu\sqrt d)$ in \eqref{eq:app-logr-gl} gives
\begin{equation}
\label{eq:app-L-representation-gl}
L_d(u)
=
\frac{u}{2}Z_d+\frac{u^2}{\mu^2}\overline B_d+R_d(u),
\qquad
|R_d(u)|\le C\frac{u^3}{\sqrt d}H_d.
\end{equation}
Hence, for fixed $M$,
$$
\sup_{0\le u\le M}
\left|L_d(u)-\left(\frac{u}{2}Z_d-\frac{u^2}{4}\right)\right|
\le
\frac{M^2}{\mu^2}\left|\overline B_d+\frac{\mu^2}{4}\right|
+C\frac{M^3}{\sqrt d}H_d
\longrightarrow0
$$
in probability by \eqref{eq:app-B2avg-gl} and \eqref{eq:app-expavg-gl}.
This proves \eqref{eq:app-boundary-gl}.
\end{proof}

\begin{lemma}[Separation away from the boundary]
\label{lem:mult-fixed-g-separation}
For every $\eta>0$ and every fixed $N<\infty$, there is $a_\eta>0$ such that
\begin{equation}
\label{eq:app-sep-gl}
\Pbb_{\theta_d}\left(
\sup_{\eta\le g\le G}
\frac1d\sum_{j=1}^d\log r_g(X_{dj})\le-a_\eta
\right)
=1-O(d^{-N}).
\end{equation}
\end{lemma}

\begin{proof}[\textbf{Proof of Lemma~\ref{lem:mult-fixed-g-separation}}]
For $Z\sim N(0,1)$, let
\[
K_0(g)=\E\log r_g(Z)
=-\operatorname{KL}\{N(0,1),m_g\},
\qquad
m_g(x)=\phi(x)r_g(x).
\]
The predictive density $m_g$ has variance $1+g\mu>1$ for $g>0$, so
$m_g\ne\phi$ and $K_0(g)<0$. Continuity on the compact interval
$[\eta,G]$ allows us to choose $a_\eta>0$ such that
\begin{equation}
\label{eq:app-KL-gap-gl}
\sup_{\eta\le g\le G}K_0(g)\le-3a_\eta.
\end{equation}

We first compare the means under the signal and null laws. Put
$F_g(a)=\E[\log r_g(Z+a)]$. If
$t_g(A)=gA/(1+gA)$ and $\E_{g,x}$ denotes expectation under weights
proportional to
\[
(1+gA)^{-1/2}e^{t_g(A)x^2/2},
\]
then
\[
\partial_x^2\log r_g(x)
=\E_{g,x}t_g(A)+x^2\Var_{g,x}\{t_g(A)\}
\le C(1+x^2)
\]
uniformly for $g\in[\eta,G]$. Since $\log r_g$ is even, $F_g'(0)=0$.
Taylor's theorem, over the bounded range containing all $\theta_{dj}$,
therefore gives
\[
\sup_{\eta\le g\le G}|F_g(a)-F_g(0)|\le Ca^2.
\]
It follows that
\begin{equation}
\label{eq:app-mean-shift-gl}
\sup_{\eta\le g\le G}
\left|
\frac1d\sum_{j=1}^d\E_{\theta_{dj}}\log r_g(X_{dj})-K_0(g)
\right|
\le\frac C d\norm{\theta_d}^2
=O(d^{-1/2}).
\end{equation}

It remains to control the centered empirical process
\[
S_d(g):=\frac1d\sum_{j=1}^d
\left[\log r_g(X_{dj})-\E\log r_g(X_{dj})\right].
\]
On $[\eta,G]$,
\[
|\log r_g(x)|+|\partial_g\log r_g(x)|\le C(1+x^2).
\]
Choose a fixed grid $\{g_k\}_{k=1}^K$ whose mesh is sufficiently small.
For every fixed integer $m$, Rosenthal's inequality gives
\[
\max_{1\le k\le K}\E|S_d(g_k)|^{2m}\le C_m d^{-m}.
\]
The same moment bound applied to the average Lipschitz envelope shows that,
outside an event of probability $O(d^{-m})$,
\[
|S_d(g)-S_d(g_k)|\le a_\eta/2
\]
whenever $g_k$ is the nearest grid point. Taking $m>N$ and applying the
union bound gives
\begin{equation}
\label{eq:app-unif-conc-gl}
\Pbb\left(\sup_{\eta\le g\le G}|S_d(g)|>a_\eta\right)
=O(d^{-N}).
\end{equation}
For all large $d$, \eqref{eq:app-KL-gap-gl} and
\eqref{eq:app-mean-shift-gl} imply that the mean empirical log likelihood is
at most $-2a_\eta$ uniformly in $g$. On the event in
\eqref{eq:app-unif-conc-gl}, it is therefore at most $-a_\eta$. This proves
\eqref{eq:app-sep-gl}.
\end{proof}

\begin{proof}[\textbf{Proof of Proposition~\ref{prop:bounded-mult-transfer}}]
Put
\[
b_0=\frac c2,
\qquad
U_d=\mu\sqrt d\,g,
\qquad
V_d=\E[U_d\mid X_d],
\qquad
\widetilde\delta_d(X_d)=\frac{V_d}{\sqrt d}X_d.
\]
We first identify the posterior of $U_d$, then compare the full posterior
mean with the radial rule $\widetilde\delta_d$, and finally evaluate its
risk.

\paragraph{Posterior localization and moments.}
After the change of variables $u=\mu\sqrt d\,g$, the posterior density of
$U_d$ is proportional to
\begin{equation}
\label{eq:app-U-posterior-gl}
u^{b_0-1}\ell\!\left(\frac{u}{\mu\sqrt d}\right)e^{L_d(u)}
\mathbf 1_{(0,\mu G\sqrt d)}(u),
\qquad
L_d(u)=\sum_{j=1}^d\log r_{u/(\mu\sqrt d)}(X_{dj}).
\end{equation}
We record a high-probability event on which the local approximation is
uniform enough for both posterior convergence and moment bounds. Fix
$\kappa=1/12$ and $\alpha=1/4$. After decreasing $\eta_0$ in
Lemma~\ref{lem:mult-local-likelihood}, choose
$K>\sup_d\E H_d+1$ and $0<\eta_1\le\eta_0$ such that
\[
\ell(g)\ge \ell(0)/2\quad(0\le g\le\eta_1),
\qquad
CK\eta_1\mu\le\frac1{16}.
\]
Let $\mathcal E_d$ be the intersection of
\[
|Z_d|\le d^\kappa,
\qquad
\left|\overline B_d+\frac{\mu^2}{4}\right|\le d^{-\alpha},
\qquad
H_d\le K,
\]
and the separation event in Lemma~\ref{lem:mult-fixed-g-separation} with
$\eta=\eta_1$. The fixed-order moment bounds used in the proofs of the two
lemmas, followed by Markov's inequality, give
\begin{equation}
\label{eq:app-good-event-gl}
\Pbb_{\theta_d}(\mathcal E_d^c)=O(d^{-6}).
\end{equation}
Indeed, the thresholds $d^\kappa$ and $d^{-\alpha}$ are polynomial, so one
only needs to take a sufficiently high fixed moment in Rosenthal's
inequality.

On $\mathcal E_d$, representation~\eqref{eq:app-L-representation-gl}
gives, for all large $d$,
\begin{equation}
\label{eq:app-tail-gl}
L_d(u)\le \frac{u}{2}Z_d-\frac18u^2,
\qquad 0\le u\le\eta_1\mu\sqrt d.
\end{equation}
Moreover, for every fixed $K_*<\infty$,
\begin{equation}
\sup_{0\le u\le K_*(1+Z_d^+)}
\left|L_d(u)-\left(\frac{u}{2}Z_d-\frac{u^2}{4}\right)\right| \le C\left\{d^{-\alpha}(1+d^{2\kappa})
+d^{-1/2}(1+d^{3\kappa})\right\}=o(1).
\label{eq:app-moving-gl}
\end{equation}
The displayed bound follows directly by substituting
$u\le K_*(1+Z_d^+)$ into~\eqref{eq:app-L-representation-gl}; our choices of
$\kappa$ and $\alpha$ make both powers of $d$ negative.

For $m=0,1,2$, write
\[
\mathcal N_{d,m}:=
\int_0^{\mu G\sqrt d}
u^{b_0-1+m}\ell\!\left(\frac{u}{\mu\sqrt d}\right)
e^{L_d(u)}\,du.
\]
Thus $V_d=\mathcal N_{d,1}/\mathcal N_{d,0}$ and
$\E[U_d^2\mid X_d]=\mathcal N_{d,2}/\mathcal N_{d,0}$.
To identify the first ratio, fix $M<\infty$ and work on
$\mathcal E_d\cap\{|Z_d|\le M\}$. For a fixed truncation point $R$,
Lemma~\ref{lem:mult-local-likelihood} and continuity of $\ell$ imply
\[
\int_0^R u^{b_0-1+m}\ell\!\left(\frac{u}{\mu\sqrt d}\right)e^{L_d(u)}du
-
\ell(0)\int_0^R u^{b_0-1+m}e^{Z_du/2-u^2/4}du
\to0
\]
in probability, for $m=0,1$. The two remaining ranges satisfy
\begin{align*}
\int_R^{\eta_1\mu\sqrt d}
 u^{b_0-1+m}\ell\!\left(\frac{u}{\mu\sqrt d}\right)e^{L_d(u)}du
&\le C\int_R^\infty u^{b_0-1+m}e^{Mu/2-u^2/8}du,\\
\int_{\eta_1\mu\sqrt d}^{\mu G\sqrt d}
 u^{b_0-1+m}\ell\!\left(\frac{u}{\mu\sqrt d}\right)e^{L_d(u)}du
&\le C d^{(b_0+m)/2}e^{-a_{\eta_1}d},
\end{align*}
by~\eqref{eq:app-tail-gl} and Lemma~\ref{lem:mult-fixed-g-separation},
respectively. The denominator is bounded away from zero on this event by
integrating over $0\le u\le1$. Letting first $d\to\infty$ and then
$R\to\infty$ gives, uniformly for $|Z_d|\le M$,
\[
V_d-
\frac{\int_0^\infty u^{b_0}e^{Z_du/2-u^2/4}\,du}
{\int_0^\infty u^{b_0-1}e^{Z_du/2-u^2/4}\,du}
\to0.
\]
Since $\{Z_d\}$ is tight, the restriction $|Z_d|\le M$ can be removed by
letting $M\to\infty$. Hence
\begin{equation}
\label{eq:app-V-gl}
V_d-h_c(Z_d)\to0
\qquad\text{in probability}.
\end{equation}

We next obtain the moment control needed to pass from
\eqref{eq:app-V-gl} to risk. Let
\[
I_b(z)=\int_0^\infty u^{b-1}e^{zu/2-u^2/4}\,du.
\]
For every $b>0$, there are $K_*>1$ and $C<\infty$ such that, uniformly in
$z\in\mathbb R$,
\begin{align}
\int_0^{K_*(1+z_+)}u^{b-1}e^{zu/2-u^2/4}\,du
&\ge C^{-1}I_b(z),
\label{eq:app-integral-lower-gl}\\
\int_{K_*(1+z_+)}^\infty u^{b+1}e^{zu/2-u^2/8}\,du
&\le C I_b(z).
\label{eq:app-integral-tail-gl}
\end{align}
For $z\ge0$ these inequalities follow by completing the square and taking
$K_*$ beyond the mode; for $z<0$ they follow after the scaling
$v=(1+z_-)u$.

Take $b=b_0$. On $\mathcal E_d$, the moving approximation and
\eqref{eq:app-integral-lower-gl} give
\begin{equation}
\label{eq:app-denominator-lower-gl}
\mathcal N_{d,0}\ge C^{-1}I_{b_0}(Z_d).
\end{equation}
To bound $\mathcal N_{d,2}$, split its integral at
$K_*(1+Z_d^+)$ and $\eta_1\mu\sqrt d$. The three pieces are bounded by
\[
C I_{b_0+2}(Z_d),
\qquad
C I_{b_0}(Z_d),
\qquad
C d^{(b_0+2)/2}e^{-a_{\eta_1}d},
\]
using~\eqref{eq:app-moving-gl},
\eqref{eq:app-tail-gl}--\eqref{eq:app-integral-tail-gl}, and fixed-$g$
separation, respectively. Since
$I_{b_0}(z)\ge C(1+z_-)^{-b_0}$, the last term is negligible after division
by~\eqref{eq:app-denominator-lower-gl}. Therefore, on $\mathcal E_d$,
\begin{align}
\E[U_d^2\mid X_d]
&\le C\left\{1+\frac{I_{b_0+2}(Z_d)}{I_{b_0}(Z_d)}\right\}\notag\\
&=C\{1+c+Z_dh_c(Z_d)\}
\le C\{1+(Z_d^+)^2\},
\label{eq:app-postmom-pointwise-gl}
\end{align}
where the identity follows by integration by parts and the last inequality
uses Lemma~\ref{lem:profile}(ii). On $\mathcal E_d^c$, the support bound
$g\le G$ gives $\E[U_d^2\mid X_d]\le\mu^2G^2d$. Consequently,
\begin{align*}
\E\left[\{\E[U_d^2\mid X_d]\}^4\right]
&\le C\E\left[1+(Z_d^+)^8\right]
+C d^4\Pbb(\mathcal E_d^c),
\end{align*}
so~\eqref{eq:app-Zmom-gl} and~\eqref{eq:app-good-event-gl} yield
\begin{equation}
\label{eq:app-postmom-L4-gl}
\sup_d\E_{\theta_d}
\left[\{\E[U_d^2\mid X_d]\}^4\right]<\infty.
\end{equation}
Jensen's
inequality consequently yields
\begin{equation}
\label{eq:app-V-L8-gl}
\sup_d\E|V_d|^8<\infty.
\end{equation}
The same denominator bound and fixed-$g$ separation give, on
$\mathcal E_d$,
\begin{equation}
\label{eq:app-posterior-tail-gl}
\Pi_d(g>\eta_1\mid X_d)
\le
\frac{C d^{b_0/2}e^{-a_{\eta_1}d}}{I_{b_0}(Z_d)}
\le C d^C e^{-ad}
\end{equation}
for some $a>0$, where the last inequality uses
$|Z_d|\le d^\kappa$ and $I_{b_0}(z)\ge C(1+z_-)^{-b_0}$.

\paragraph{Asymptotic radiality.}
By~\eqref{eq:app-qg-exp}, for $g\le\eta_1$,
\[
|q_g(X_{dj})-\mu g|
\le Cg^2(1+X_{dj}^2)e^{C\eta_1X_{dj}^2}.
\]
For $g>\eta_1$, both $q_g(X_{dj})$ and $\mu g$ are uniformly bounded.
Using~\eqref{eq:app-exact-postmean-gl} and splitting the posterior at
$\eta_1$ therefore gives, on $\mathcal E_d$,
\begin{equation}
|\delta_{dj}-\widetilde\delta_{dj}|
\le \frac{C}{d}|X_{dj}|(1+X_{dj}^2)e^{C\eta_1X_{dj}^2}
\ \E[U_d^2\mid X_d] +C|X_{dj}|\Pi_d(g>\eta_1\mid X_d).
\label{eq:app-coordinate-difference-gl}
\end{equation}
After decreasing $\eta_1$ if necessary, boundedness of the coordinates of
$\theta_d$ and Gaussian exponential moments imply
\begin{equation}
\label{eq:app-expavg-L2-gl}
\sup_d\E\left[
\left\{\frac1d\sum_{j=1}^d
X_{dj}^2(1+X_{dj}^2)^2e^{2C\eta_1X_{dj}^2}\right\}^2
\right]<\infty.
\end{equation}
Squaring the first term in~\eqref{eq:app-coordinate-difference-gl}, summing
over $j$, and applying Cauchy--Schwarz gives
\begin{align*}
&\frac{C}{d^2}\E\left[
\{\E[U_d^2\mid X_d]\}^2
\sum_{j=1}^d X_{dj}^2(1+X_{dj}^2)^2e^{2C\eta_1X_{dj}^2}
\right]\\
&\qquad\le \frac{C}{d^2}
\left[\E\{\E[U_d^2\mid X_d]\}^4\right]^{1/2}
\left[\E\left\{\sum_{j=1}^d
X_{dj}^2(1+X_{dj}^2)^2e^{2C\eta_1X_{dj}^2}\right\}^2\right]^{1/2}
=O(d^{-1})
\end{align*}
by~\eqref{eq:app-postmom-L4-gl} and~\eqref{eq:app-expavg-L2-gl}.
The second term is $o(d^{-1})$ by
\eqref{eq:app-posterior-tail-gl} and $T_d=O(d)$ on $\mathcal E_d$. Finally,
on $\mathcal E_d^c$ the squared discrepancy is bounded by $CT_d$, so
\[
\E[T_d\mathbf 1_{\mathcal E_d^c}]
\le (\E T_d^2)^{1/2}\Pbb(\mathcal E_d^c)^{1/2}=o(d^{-1}).
\]
Thus
\begin{equation}
\label{eq:app-nonradial-gl}
\E_{\theta_d}\norm{\delta_d-\widetilde\delta_d}^2=O(d^{-1}).
\end{equation}

\paragraph{Risk limit.}
Equation~\eqref{eq:app-V-gl}, continuity of $h_c$, and
$Z_d\Rightarrow Z_\beta$ give
$V_d\Rightarrow h_c(Z_\beta)$. The moment bound
\eqref{eq:app-V-L8-gl} supplies uniform integrability, and
$T_d/d=1+Z_d/\sqrt d\to1$ in probability. Therefore
\begin{equation}
\label{eq:app-radial-norm-gl}
\E\norm{\widetilde\delta_d}^2
=\E\left[V_d^2\frac{T_d}{d}\right]
\longrightarrow \E\left[h_c(Z_\beta)^2\right].
\end{equation}
Moreover,
\[
\theta_d^\top\widetilde\delta_d
=\frac{V_d}{\sqrt d}\norm{\theta_d}^2
+\frac{V_d}{\sqrt d}\theta_d^\top\varepsilon_d.
\]
The first term converges in expectation to
$\beta\E h_c(Z_\beta)$, while Cauchy--Schwarz and
\eqref{eq:app-V-L8-gl} give
\[
\left|\E\left(\frac{V_d}{\sqrt d}\theta_d^\top\varepsilon_d\right)\right|
\le \frac{(\E V_d^2)^{1/2}}{\sqrt d}\norm{\theta_d}=o(1).
\]
Consequently,
\[
\risk(\theta_d,\widetilde\delta_d)-\norm{\theta_d}^2
\longrightarrow
\E\left[h_c(Z_\beta)^2\right]-2\beta\E h_c(Z_\beta)
=L_c(\beta).
\]
Letting $e_d=\delta_d-\widetilde\delta_d$ and using
\eqref{eq:app-nonradial-gl},
\begin{align*}
|\risk(\theta_d,\delta_d)-\risk(\theta_d,\widetilde\delta_d)|
&\le \E\norm{e_d}^2
+2\{\E\norm{e_d}^2\}^{1/2}
\{\risk(\theta_d,\widetilde\delta_d)\}^{1/2}\\
&=o(1),
\end{align*}
because $\risk(\theta_d,\widetilde\delta_d)=O(\sqrt d)$. This proves the
critical-regime assertion. If $\norm{\theta_d}^2\to\nu$, then $\beta=0$
and Lemma~\ref{lem:profile}(iii) gives $L_c(0)=c$, so
$\risk(\theta_d,\delta_d)\to\nu+c$.
\end{proof}

\subsubsection{Auxiliary calculations for priors beyond a single global scale}

This subsection provides the auxiliary calculations for Section~\ref{subsec:beyondglobal}. Proposition~\ref{prop:bounded-mult-transfer} gives exact weak-signal and critical-regime transfer for bounded coordinate multipliers. The calculations below use that case as a benchmark and then identify how unbounded local tails, changing model-size probabilities, and dimension-dependent allocation vectors modify the relevant near-zero problem.

The bounded-multiplier proof has two ingredients. First, the marginal one-coordinate predictive ratio satisfies $\log r_g(x)=\mu\alpha_1(x)g+B_2(x)g^2+\rho_g(x)$, where $\mu=\E A$ and $\rho_g$ obeys the uniform bound in \eqref{eq:app-rho-gl}. Second, after summing over coordinates, the posterior for the common variance is governed by the local variable $u=\mu\sqrt d\,g$. This is what makes the multiplier distribution disappear from the limiting experiment. Priors outside the theorem fail to fit this template for one of three reasons: the local multiplier may have too much tail mass, the model-size probability may change with $d$ or be learned from the data, or a dimension-dependent allocation vector may distribute the common scale unevenly across coordinates.

\paragraph{Light-tailed unbounded local scales.}
Let $A\ge0$ denote the local variance multiplier in a global--local normal mixture. If $\E A<\infty$, then
\[
\frac{1}{g}\E\left[\frac{gA}{1+gA}\right]
=
\E\left[\frac{A}{1+gA}\right]
\to \E A,
\quad g\downarrow0.
\]
Thus the first-order effective scale is $g\E A+o(g)$, exactly the same scale as in the bounded-multiplier argument. This supports the same common-scale classification: a prior with positive finite density on $\tau$ has the SD-scale geometry, while a prior with positive finite density on $g=\tau^2$ has the variance-scale geometry. What remains outside the present theorem is not the first-order scale calculation, but the uniform posterior-risk control. With unbounded $A$, very large local multipliers can interact with large observations in the likelihood tails; a full theorem would need conditions that make those tail terms negligible uniformly over the critical-regime experiment.

\paragraph{Horseshoe-type local scales.}

For the horseshoe local-scale distribution, $A=\lambda^2$ with $\lambda\sim C^+(0,1)$, the finite-mean calculation is no longer valid because $\E A=\infty$. The relevant small-scale average can nevertheless be computed exactly:
\[
\E\left[\frac{g\lambda^2}{1+g\lambda^2}\right]
=
\frac{2}{\pi}\int_0^\infty
\frac{g\lambda^2}{1+g\lambda^2}\frac{1}{1+\lambda^2}\,d\lambda
=
\frac{\sqrt g}{1+\sqrt g}
=
\frac{\tau}{1+\tau}.
\]
Equivalently, if $K_j=(1+\tau^2\lambda_j^2)^{-1}$ is the usual shrinkage factor, then
\[
\E[1-K_j\mid \tau]=\frac{\tau}{1+\tau}\sim \tau,
\quad \tau\downarrow0.
\]
Hence the heavy local tail moves the near-zero scale from the variance coordinate $g$ to the SD coordinate $\tau$. This is why a half-Cauchy or half-$t$ prior on the global $\tau$ is aligned with the SD-scale geometry. At the same time, the local Cauchy tail is part of the leading local near-zero problem, so the exact bounded-multiplier constants $L_c(\beta)$ are not claimed for the full horseshoe risk.

This distinction also clarifies scale calibration. The exponent of $p(\tau)$ near zero specifies the local shape of the common scale, while the numerical width of that prior controls the prior expected amount of nonnegligible shrinkage. For the horseshoe, the prior mean of the total effective nonzero mass satisfies
\[
\sum_{j=1}^d \E[1-K_j\mid \tau]
=
 d\frac{\tau}{1+\tau}
\sim d\tau,
\quad \tau\downarrow0,
\]
so calibration of the width of the global $\tau$ prior is a separate sparsity-calibration problem.

\paragraph{Spike-and-slab priors.}

The fixed-$q$ point-mass spike-and-slab is already covered by Proposition~\ref{prop:bounded-mult-transfer} and its proof in Appendix~\ref{app:coord-mult-proof}, through the specialization $A_j=Z_j\sim\operatorname{Bernoulli}(q)$. The case requiring a separate calculation is the sparse regime in which the model-size probability is no longer fixed. If $q=q_d\to0$, the one-coordinate marginal ratio is
\[
r_g^{(q_d)}(x)
=
1+q_d\left\{(1+g)^{-1/2}
\exp\left(\frac{x^2g}{2(1+g)}\right)-1\right\}.
\]

For $0\le g\le\eta$ with $\eta>0$ sufficiently small, put
\[
s_g(x)=(1+g)^{-1/2}
\exp\left\{\frac{x^2g}{2(1+g)}\right\}.
\]
A second-order Taylor expansion in $g$, with the derivatives controlled by a
Gaussian exponential envelope, gives uniformly in $x$,
\[
|s_g(x)-1-g(x^2-1)/2|
\le Cg^2(1+x^4)e^{Cgx^2},
\qquad
|s_g(x)-1|
\le Cg(1+x^2)e^{Cgx^2}.
\]
Since $r_g^{(q_d)}(x)=1+q_d\{s_g(x)-1\}$ and
$|\log(1+y)-y|\le Cy^2$ uniformly for
$y\ge-\{1-(1+\eta)^{-1/2}\}$, it follows that
\begin{align*}
\log r_g^{(q_d)}(x)
&=
\frac{q_dg}{2}(x^2-1)+R_{d,g}(x),\\
|R_{d,g}(x)|
&\le
Cq_dg^2(1+x^4)e^{Cgx^2}
+Cq_d^2g^2(1+x^2)^2e^{2Cgx^2},
\end{align*}
uniformly in $x$, $q_d\in[0,1]$, and $0\le g\le\eta$.

Accordingly, the near-zero scale contains the product $q_dg$. The slab-scale exponent alone cannot determine leading risk when $q_d$ is triangular, assigned a hyperprior, or estimated empirically: the model-size prior changes how much of the common slab scale is actually exposed to the $d$ coordinates.

\paragraph{Dirichlet--Laplace and R2-D2 priors.}

Dirichlet--Laplace priors separate a common total scale from a dimension-dependent allocation vector. In the usual normal-means normalization,
\[
\tau\sim\operatorname{Gamma}(d a_d,1/2),
\]
so
\[
p_d(\tau)\propto \tau^{d a_d-1}
\quad (\tau\downarrow0).
\]
The common-scale index is therefore $c_d=d a_d$; the common choice $a_d=1/d$ gives the SD-scale geometry $c=1$. This statement classifies only the common total scale. The Dirichlet weights $(\phi_1,\ldots,\phi_d)$ determine how that total scale is allocated across coordinates and therefore form an additional dimension-dependent part of the risk problem.

For R2-D2 and related $R^2$-based shrinkage priors, apply Proposition~\ref{prop:commonhyperapp}(iii) to the common scale $\gamma$, with $a=A$ and $b=B$. If $R^2=w_0\gamma^2/(1+w_0\gamma^2)$ and $R^2\sim\operatorname{Beta}(A,B)$, then
\[
p(\gamma)
=
\frac{2w_0^A}{\mathrm B(A,B)}\gamma^{2A-1}(1+w_0\gamma^2)^{-(A+B)},
\]
so the common-scale index is $c=2A$. Hence $A=1/2$ gives the SD-scale class and $A=1$ gives the variance-scale class. As with Dirichlet--Laplace, this classification does not remove the need to analyze the allocation prior or model-fit component.

\section{Further structure of the critical-regime gap}
\label{app:shellgap}
This appendix records several additional facts about
\[
\Delta(\beta):=L_1(\beta)-L_2(\beta),
\]
the critical-regime comparison between the SD-flat rule and the harmonic benchmark. These facts sharpen the interpretation of the critical regime. In particular, they yield the upper-tail expansion for $\Delta(\beta)$ used in the main text and show that the harmonic benchmark has smaller centered risk for all sufficiently large signals in the critical regime.

Lemma~\ref{lem:profile}(iii) gives, for every $c>0$ and every $\beta\ge0$,
\[
q_c(\beta)=c+\beta m_c(\beta),
\qquad
L_c(\beta)=c-\beta m_c(\beta).
\]
Consequently,
\[
\Delta(\beta)=-1+\beta\,\mathcal D(\beta),
\qquad
\mathcal D(\beta):=m_2(\beta)-m_1(\beta).
\]

\begin{proposition}[Exact simplification of the critical-regime constants]
\label{prop:deltasimplify}
For every $c>0$ and every $\beta\ge 0$,
\[
q_c(\beta)=c+\beta m_c(\beta)
\quad\text{and hence}\quad
L_c(\beta)=c-\beta m_c(\beta).
\]
Therefore
\[
\Delta(\beta)=-1+\beta\,\mathcal D(\beta),
\quad
\mathcal D(\beta):=m_2(\beta)-m_1(\beta).
\]
\end{proposition}

\begin{proof}[\textbf{Proof of Proposition~\ref{prop:deltasimplify}}]
Equation \eqref{eq:ricatti_id} from Lemma~\ref{lem:profile} gives
\[
h_c'(z)=\frac c2+\frac z2 h_c(z)-\frac12 h_c(z)^2.
\]
Since $Z_\beta\sim N(\beta,2)$, Stein's identity yields
\[
\E[(Z_\beta-\beta)h_c(Z_\beta)]=2\E[h_c'(Z_\beta)].
\]
Substituting the Riccati formula into the right-hand side gives
\[
\E[Z_\beta h_c(Z_\beta)]-\beta m_c(\beta)
=
c+\E[Z_\beta h_c(Z_\beta)]-q_c(\beta),
\]
which simplifies to
\[
q_c(\beta)=c+\beta m_c(\beta).
\]
The remaining identities are immediate.
\end{proof}

\begin{proposition}[Pointwise ordering]
\label{prop:pointwisegap}
For every $z\in\R$,
\[
h_2(z)>h_1(z).
\]
Consequently,
\[
\mathcal D(\beta)>0
\quad\text{for every }\beta\ge 0.
\]
\end{proposition}

\begin{proof}[\textbf{Proof of Proposition~\ref{prop:pointwisegap}}]
For fixed $z$, let
\[
q_{c,z}(u)\propto u^{c/2-1}\exp\!\left(\frac{zu}{2}-\frac{u^2}{4}\right),
\quad u>0,
\]
so that $h_c(z)=\E_{q_{c,z}}[U]$. For $c=1$ and $c=2$,
\[
\frac{q_{2,z}(u)}{q_{1,z}(u)}\propto u^{1/2},
\]
which is strictly increasing in $u$. Thus $q_{2,z}$ dominates $q_{1,z}$ in monotone likelihood
ratio order, and the expectation of the increasing function $u\mapsto u$ is strictly larger under
$q_{2,z}$ than under $q_{1,z}$. Averaging over $Z_\beta$ gives $\mathcal D(\beta)>0$.
\end{proof}

\begin{proposition}[Large-$z$ expansion of the critical-regime profiles]
\label{prop:hlarge}
For every fixed $c>0$,
\[
h_c(z)
=
z+\frac{c-2}{z}+\frac{(c-2)(4-c)}{z^3}+O(z^{-5}),
\quad z\to\infty.
\]
In particular,
\[
h_2(z)-h_1(z)=z^{-1}+3z^{-3}+O(z^{-5}).
\]
\end{proposition}

\begin{proof}[\textbf{Proof of Proposition~\ref{prop:hlarge}}]
Write $a:=c/2$ and recall
\[
I_a(z)=\int_0^\infty u^{a-1}\exp\!\left(\frac{zu}{2}-\frac{u^2}{4}\right)\,du,
\quad
h_c(z)=\frac{I_{a+1}(z)}{I_a(z)}.
\]
Completing the square gives
\[
I_a(z)=e^{z^2/4}\int_0^\infty u^{a-1}e^{-(u-z)^2/4}\,du.
\]
Split the integral at $u=z/2$. On $(0,z/2)$,
\[
e^{-(u-z)^2/4}\le e^{-z^2/16},
\]
so
\[
\int_0^{z/2}u^{a-1}e^{-(u-z)^2/4}\,du = O\!\left(z^a e^{-z^2/16}\right).
\]
Thus the lower piece is exponentially negligible.

On $[z/2,\infty)$, write $u=z+w$. Then
\[
I_a(z)
=
e^{z^2/4} z^{a-1}\int_{-z/2}^{\infty}
\left(1+\frac{w}{z}\right)^{a-1}e^{-w^2/4}\,dw.
\]
Fix $\eta\in(1/2,1)$. The contribution of $|w|>z^\eta$ is exponentially small, so it suffices to
consider $|w|\le z^\eta$, where $|w|/z\le z^{\eta-1}\to 0$. A fifth-order Taylor expansion yields
\[
\left(1+\frac{w}{z}\right)^{a-1}
=
\sum_{k=0}^5 \binom{a-1}{k}\left(\frac{w}{z}\right)^k
+
O\!\left(\frac{|w|^6}{z^6}\right),
\]
uniformly on $|w|\le z^\eta$. Because the omitted tails are exponentially small, we may integrate
termwise over $\mathbb R$ and use the Gaussian moments
\[
\int_{\mathbb R}e^{-w^2/4}\,dw=2\sqrt\pi,
\quad
\int_{\mathbb R}w^2 e^{-w^2/4}\,dw=4\sqrt\pi,
\quad
\int_{\mathbb R}w^4 e^{-w^2/4}\,dw=24\sqrt\pi.
\]
The odd moments vanish, so
\[
I_a(z)
=
2\sqrt\pi\, e^{z^2/4} z^{a-1}
\left[
1+\frac{(a-1)(a-2)}{z^2}
+\frac{(a-1)(a-2)(a-3)(a-4)}{2z^4}
+O(z^{-6})
\right].
\]
Replacing $a$ by $a+1$ gives
\[
I_{a+1}(z)
=
2\sqrt\pi\, e^{z^2/4} z^{a}
\left[
1+\frac{a(a-1)}{z^2}
+\frac{a(a-1)(a-2)(a-3)}{2z^4}
+O(z^{-6})
\right].
\]
Taking the ratio and simplifying the coefficients yields
\[
h_c(z)
=
z+\frac{c-2}{z}+\frac{(c-2)(4-c)}{z^3}+O(z^{-5}),
\]
as claimed. The formula for $h_2-h_1$ is immediate.
\end{proof}

\begin{proof}[\textbf{Proof of Corollary~\ref{cor:deltatail}}]
The risk-gap convergence follows by subtracting the two centered limits in Theorem~\ref{thm:main}. By Lemma~\ref{lem:profile}(iii), $L_c(0)=c$, and therefore $\Delta(0)=L_1(0)-L_2(0)=-1$.

Let $Z_\beta=\beta+W$ with $W\sim N(0,2)$, and let
\[
A_\beta:=\{|W|\le \beta/2\}.
\]
By Proposition~\ref{prop:hlarge},
\[
h_2(z)-h_1(z)=z^{-1}+3z^{-3}+O(z^{-5}),
\quad z\to\infty.
\]
On $A_\beta$, $Z_\beta\in[\beta/2,3\beta/2]$, so
\[
h_2(Z_\beta)-h_1(Z_\beta)
=
Z_\beta^{-1}+3Z_\beta^{-3}+O(\beta^{-5}).
\]
Also, Lemma~\ref{lem:profile}(ii) implies
\[
|h_2(z)-h_1(z)|\le C(1+|z|),
\]
so the complement contributes only exponentially small mass:
\[
\E\!\left[|h_2(Z_\beta)-h_1(Z_\beta)|\mathbf 1_{A_\beta^c}\right]
=
O(e^{-c_0\beta^2})
\]
for some $c_0>0$.

On $A_\beta$, Taylor expansion in $W/\beta$ gives
\[
(\beta+W)^{-1}
=
\beta^{-1}-\beta^{-2}W+\beta^{-3}W^2
-\beta^{-4}W^3+O(\beta^{-5}|W|^4),
\]
and
\[
(\beta+W)^{-3}
=
\beta^{-3}-3\beta^{-4}W+O(\beta^{-5}W^2).
\]
The event $A_\beta$ and the law of $W$ are symmetric, so the odd terms have zero truncated expectation. Using $\E[W^2]=2$ and $\Pbb(A_\beta^c)=O(e^{-c_0\beta^2})$ therefore gives
\[
\E[Z_\beta^{-1}\mathbf 1_{A_\beta}]
=
\beta^{-1}+2\beta^{-3}+O(\beta^{-5}),
\quad
\E[Z_\beta^{-3}\mathbf 1_{A_\beta}]
=
\beta^{-3}+O(\beta^{-5}).
\]
Therefore
\[
\mathcal D(\beta)
=
\E[h_2(Z_\beta)-h_1(Z_\beta)]
=
\beta^{-1}+5\beta^{-3}+O(\beta^{-5}).
\]
Now Lemma~\ref{lem:profile}(iii) gives
\[
\Delta(\beta)=-1+\beta\,\mathcal D(\beta)=5\beta^{-2}+O(\beta^{-4}),
\]
so $\Delta(\beta)>0$ for all sufficiently large $\beta$.

Finally, continuity of $\Delta$ follows from dominated convergence and the linear-growth bound from
Lemma~\ref{lem:profile}(ii). Since $\Delta(0)=-1$, the intermediate value
theorem yields at least one zero $\beta_*\in(0,\infty)$.
\end{proof}

\end{document}

%% file: sim_tables/tab_regime_check_main_20260507.tex
\begin{table}[t!]
\caption{Selected exact finite-$d$ checks for the three asymptotic regimes}
\label{tab:regime-check-main}
\begin{center}
\small
\begin{tabular}{cccccc}
\toprule
Regime & $d$ & calibration & $c$ & finite-$d$ value & asymptotic target \\
\midrule
\multicolumn{6}{l}{\textit{Weak:} $R_d(\lambda;c)$} \\
weak & 2000 & $\lambda=0$ & 1 & 1.000 & 1.000 \\
weak & 2000 & $\lambda=0$ & 2 & 2.000 & 2.000 \\
weak & 2000 & $\lambda=1$ & 1 & 1.981 & 2.000 \\
weak & 2000 & $\lambda=1$ & 2 & 2.971 & 3.000 \\
weak & 2000 & $\lambda=4$ & 1 & 4.922 & 5.000 \\
weak & 2000 & $\lambda=4$ & 2 & 5.882 & 6.000 \\
\midrule
\multicolumn{6}{l}{\textit{Critical:} $R_d(\beta\sqrt d;c)-\beta\sqrt d$} \\
critical & 5000 & $\beta=1.0$ & 1 & -0.272 & -0.274 \\
critical & 5000 & $\beta=1.0$ & 2 & 0.250 & 0.244 \\
critical & 5000 & $\beta=\beta_*\approx 2.08$ & 1 & -3.034 & -3.105 \\
critical & 5000 & $\beta=\beta_*\approx 2.08$ & 2 & -3.018 & -3.105 \\
critical & 5000 & $\beta=3.0$ & 1 & -6.952 & -7.219 \\
critical & 5000 & $\beta=3.0$ & 2 & -7.229 & -7.521 \\
\midrule
\multicolumn{6}{l}{\textit{Strong:} $R_d(\rho d;c)/d$} \\
strong & 200 & $\rho=0.25$ & 1 & 0.208 & 0.208 \\
strong & 200 & $\rho=0.25$ & 2 & 0.206 & 0.208 \\
strong & 200 & $\rho=1$ & 1 & 0.506 & 0.506 \\
strong & 200 & $\rho=1$ & 2 & 0.506 & 0.506 \\
strong & 200 & $\rho=4$ & 1 & 0.803 & 0.803 \\
strong & 200 & $\rho=4$ & 2 & 0.803 & 0.803 \\
\bottomrule
\end{tabular}
\end{center}

\setlength{\baselineskip}{4mm}
\end{table}

%% file: sim_tables/tab_transfer_priors_main.tex
\begin{table}[t!]
\caption{Representative single global-scale hyperpriors at $d=5000$}
\label{tab:transfer-priors-main}
\begin{center}
\small
\begin{tabular}{lccc}
\toprule
Prior & $R_d(0;p)$ & $R_d(\sqrt d;p)-\sqrt d$ & $R_d(3\sqrt d;p)-3\sqrt d$ \\
\midrule
\multicolumn{4}{l}{\textit{$c=1$ class target: } $1$, $L_1(1)=-0.274$, $L_1(3)=-7.219$} \\
Half-Cauchy$(1)$ & 0.977 & -0.292 & -6.952 \\
Half-normal$(1)$ & 0.988 & -0.282 & -6.953 \\
Truncated flat on $\tau$ & 1.000 & -0.272 & -6.952 \\
Beta-$R^2$ $(a=\tfrac12,b=1)$ & 0.965 & -0.302 & -6.952 \\
\midrule
\multicolumn{4}{l}{\textit{$c=2$ class target: } $2$, $L_2(1)=0.244$, $L_2(3)=-7.521$} \\
Gamma on $g$ $(a=1,b=1)$ & 1.963 & 0.215 & -7.248 \\
Beta-$R^2$ $(a=1,b=1)$ & 1.929 & 0.184 & -7.263 \\
Truncated flat on $g$ & 2.000 & 0.250 & -7.229 \\
\bottomrule
\end{tabular}
\end{center}
{\footnotesize {\em Notes}: The displayed values are exact finite-\(d\) checks. In the second and third risk columns, the arguments \(\sqrt d\) and \(3\sqrt d\) refer to signal energy \(\norm{\theta_d}^2\). Within each class, the near-zero exponent of the SD-scale density determines the asymptotic target.}
\setlength{\baselineskip}{4mm}
\end{table}

%% file: sim_tables/tab_beyond_global_twoarch_20260507.tex
\begin{table}[t!]
\caption{Many-weak weak-signal comparison for two architectures beyond a single global scale}
\label{tab:beyond-global-twoarch}
\begin{center}
\small
\begin{tabular}{ccccccc}
\toprule
 & \multicolumn{3}{c}{Global-local (half-Cauchy locals)} & \multicolumn{3}{c}{Spike-and-slab ($q_d=m_d/d$)} \\
\cmidrule(lr){2-4}\cmidrule(lr){5-7}
$\lambda$ & $c=1$ & $c=2$ & Diff & $c=1$ & $c=2$ & Diff \\
\midrule
0 & 1.01 $(0.07)$ & 1.91 $(0.10)$ & 0.90 & 0.72 $(0.04)$ & 1.45 $(0.06)$ & 0.72 \\
1 & 0.94 $(0.07)$ & 1.82 $(0.09)$ & 0.88 & 0.68 $(0.05)$ & 1.38 $(0.06)$ & 0.69 \\
2 & 1.05 $(0.08)$ & 1.99 $(0.10)$ & 0.94 & 0.80 $(0.06)$ & 1.50 $(0.07)$ & 0.71 \\
4 & 0.97 $(0.09)$ & 1.83 $(0.11)$ & 0.86 & 0.62 $(0.05)$ & 1.26 $(0.06)$ & 0.64 \\
\bottomrule
\end{tabular}
\end{center}
{\footnotesize {\em Notes}: The table reports weak-regime excess risk \(\widehat R-\lambda\), where \(\lambda=\norm{\theta_d}^2\) denotes signal energy, under the many-weak sparse design with \(d=500\), \(m_d=\lfloor d^{3/4}\rfloor\), alternating signs on the first \(m_d\) coordinates, and \(q_d=m_d/d\) in the spike-and-slab architecture. The \(c=1\) columns use a half-Cauchy prior on the common global or slab SD scale; the \(c=2\) columns replace that prior component by an exponential prior on the corresponding variance scale. Parentheses show Monte Carlo standard errors, and ``Diff'' denotes the within-architecture increase in excess risk when the common scale changes from the \(c=1\) class to the \(c=2\) class.}
\setlength{\baselineskip}{4mm}
\end{table}